\documentclass[conference,fleqn]{IEEEtran}
\usepackage{algorithm}
\usepackage{algpseudocode}
\usepackage{amssymb}
\usepackage{amsmath}
\usepackage{amsthm}

\usepackage[ligature, inference]{semantic}
\usepackage{proof}
\usepackage{soul}

\newcommand{\unfix}[2]{}
\newtheorem{thm}{Theorem}

\newtheorem{lemma}{Lemma}
\newtheorem{definition}{Definition}

\usepackage{xcolor}
\DeclareMathOperator{\ord}{\mathit{ord}}
\DeclareMathOperator{\lt}{\mathit{lt}}

\usepackage{hyperref}

\newcommand{\algeq}{=_E}
\DeclareMathOperator{\dom}{dom}
\DeclareMathOperator{\Ana}{Ana}
\DeclareMathOperator{\Chck}{Check}

\newcommand{\role}[1]{\mathsf{#1}}

\newcommand{\inact}{\mathbf 0}
\newcommand{\lock}{\mathsf{lock}}
\newcommand{\unlock}{\mathsf{unlock}}

\newcommand{\ite}[3]{\mathsf{if}~ #1 ~\mathsf{then}~ #2 ~\mathsf{else}~ #3}

\newcommand{\sem}[1]{[\![#1]\!]}

\newcommand{\A}{\mathcal{A}}

\newcommand{\C}{\mathcal{C}}

\newcommand{\F}{\mathcal{F}}

\newcommand{\LC}{\mathcal{L}}
\newcommand{\X}{\mathcal{X}}
\newcommand{\T}{\mathcal{T}}

\newcommand{\I}{\mathcal{I}}

\newcommand{\kwl}[1]{\mathbf{#1}}
\newcommand{\kwf}[1]{\mathsf{#1}}
\newcommand{\kwc}[1]{\mathsf{#1}}

\newcommand{\kwt}[1]{\mathsf{#1}}
\newcommand{\kwe}[1]{\mathsf{#1}}

\newcommand{\var}[1]{\mathit{#1}}

\newcommand{\pvt}[1]{[\![#1]\!]_\mathsf{pv}}

\newcommand{\NO}{\mathit{NO}}
\newcommand{\NR}{\mathit{NR}}
\newcommand{\Text}{\mathit{Text}}
\newcommand{\ttp}{\role {ttp}}
\newcommand{\pk}{\mathit{pk}}
\newcommand{\ek}{\mathit{ek}}
\newcommand{\sk}{\mathit{sk}}
\newcommand{\timeout}{\mathit{timeout}}
\newcommand{\fresolve}{f_{\mathit{resolve}}}
\newcommand{\fabort}{f_{\mathit{abort}}}
\newcommand{\resolved}{\mathit{resolved}}
\newcommand{\aborted}{\mathit{aborted}}
\newcommand{\ttpmem}{\mathit{ttp\_mem}}

\newcommand{\blank}{\mathit{blank}}
\newcommand{\upd}{\mathit{upd}}
\newcommand{\msg}{\mathit{msg}}

\newcommand{\shk}{\mathit{shk}}
\newcommand{\pair}{\mathit{pair}}
\newcommand{\vpair}{\mathit{vpair}}
\newcommand{\fst}{\mathit{fst}}
\newcommand{\snd}{\mathit{snd}}
\newcommand{\scrypt}{\mathit{scrypt}}
\newcommand{\dscrypt}{\mathit{dscrypt}}
\newcommand{\vscrypt}{\mathit{vscrypt}}
\newcommand{\crypt}{\mathit{crypt}}
\newcommand{\dcrypt}{\mathit{dcrypt}}
\newcommand{\vcrypt}{\mathit{vcrypt}}
\newcommand{\sign}{\mathit{sign}}
\newcommand{\open}{\mathit{open}}
\newcommand{\vsign}{\mathit{vsign}}
\newcommand{\inv}{\mathit{inv}}
\newcommand{\pubk}{\mathit{pub_k}}
\newcommand{\vinv}{\mathit{vinv}}
\newcommand{\cexp}{\mathit{exp}}
\newcommand{\expinv}{\mathit{exp}^{-1}}
\newcommand{\vexp}{\mathit{vexp}}

\newcommand{\niauth}[3]{#1~\mathsf{ni\text{-}authenticates}~#2~\mathsf{on}~#3}
\newcommand{\auth}[3]{#1~\mathsf{authenticates}~#2~\mathsf{on}~#3}
\newcommand{\secret}[2]{#1~\mathsf{secret~between}~#2}

\newcommand{\fv}{\mathsf{fv}}
\newcommand{\send}{\mathsf{send}}
\newcommand{\receive}{\mathsf{receive}}
\newcommand{\event}{\mathsf{event}}
\newcommand{\Ag}{\mathit{Ag}}

\renewcommand{\i}{\role i}
\newcommand{\cont}{\mathit{cont}}

\newcommand{\confver}[1]{}
\newcommand{\diffver}[1]{}
\newcommand{\extver}[1]{#1}
\usepackage[defaultcolor=magenta,final]{changes}

\newcommand{\placeholder}{\url{http://imm.dtu.dk/\~samo/cryptochoreo.pdf}}

\title{Cryptographic Choreographies}

\author{\IEEEauthorblockN{Sebastian M{\"o}dersheim}
\IEEEauthorblockA{\textit{Technical University of}\\\textit{Denmark,}
Denmark\\
0000-0002-6901-8319}
\and
\IEEEauthorblockN{Simon Lund}
\IEEEauthorblockA{\textit{Technical University of}\\\textit{Denmark,}
Denmark\\
0009-0005-2957-3472}
\and
\IEEEauthorblockN{Alessandro Bruni}
\IEEEauthorblockA{\textit{IT-University of 
  }\\
  \emph{Copenhagen}, Denmark\\
0000-0003-2946-9462}
\and
\IEEEauthorblockN{Marco Carbone}
\IEEEauthorblockA{\textit{IT-University of 
  }\\
  \emph{Copenhagen}, Denmark\\
  0000-0001-9479-2632 }
\and
\IEEEauthorblockN{Rosario Giustolisi}
\IEEEauthorblockA{\textit{IT-University of 
  }\\
  \emph{Copenhagen}, Denmark\\
  0000-0002-2917-9601}
}

\begin{document}

\maketitle

\begin{abstract}
  We present CryptoChoreo, a choreography language for the
  specification of cryptographic protocols. Choreographies can be
  regarded as an extension of Alice-and-Bob notation, providing an
  intuitive high-level view of the protocol as a whole (rather than
  specifying each protocol role in isolation). The extensions over
  standard Alice-and-Bob notation that we consider are
  nondeterministic choice, conditional branching, and mutable
  long-term memory. We define the semantics of CryptoChoreo by
  translation to a process calculus. This semantics entails an
  understanding of the protocol: it determines how agents parse and
  check incoming messages and how they construct outgoing messages, in
  the presence of an arbitrary algebraic theory and nondeterministic
  choices made by other agents. While this semantics entails algebraic
  problems that are in general undecidable, we give an implementation
  for a representative theory.  We connect this translation to
  ProVerif and show on a number of case studies that the approach is
  practically feasible.
\end{abstract}

\vspace{1ex}
\noindent \textbf{Acknowledgments:} Partly funded by EU Horizon Europe under Grant Agreement no. 101093006 (TaRDIS).

\section{Introduction}\label{sec:intro}
Specification languages for security protocols can be roughly divided
into three classes. The most low-level one are based on multi-set
rewriting rules, such as the Tamarin input
language~\cite{MeierSchmidtCremersBasin13} and the AVISPA Intermediate
Format~\cite{ArmandoEtAl05}, where each rule describes state
transitions corresponding usually to a pair of protocol steps from the
view of one honest agent: receiving a message, processing and checking
it, and sending the next message. More high-level are languages based
on process calculus such as ProVerif~\cite{Blanchet16}, where
typically each role of the protocol is described like a program, often
as a sequence of sending and receiving steps. The most high-level are
languages based on Alice-and-Bob
notation~\cite{DBLP:conf/nspw/Millen96,DBLP:journals/jcs/Lowe98,JRVP00,Modersheim09,DBLP:journals/ipl/ChevalierR10,sps,DBLP:conf/birthday/BasinKRS15},
which describe the entire
protocol by an ideal run of the protocol as a sequence of
$A\rightarrow B: M$ steps where role $A$ sends message $M$ to role
$B$, thus describing the interplay of all roles.

Alice-and-Bob notation is very intuitive and succinct because it gives
the synopsis of the protocol and leaves implicit how agents construct
the messages they send, and how they parse and check the messages they
receive. The latter is a non-trivial problem one has to solve when
defining a formal language based on Alice-and-Bob notation, namely
when giving a formal \emph{semantics} by translation to a lower-level
language. This was described by~\cite{JRVP00,CaleiroBasinVigano06} for
models in the free term algebra, but a key question is how to deal
with algebraic properties as needed, for instance, for
Diffie-Hellman. If Alice needs to construct $\exp(\exp(g,X),Y)$, given
her own secret $X$ and the public value $\exp(g,Y)$ from Bob, the
semantics needs to infer that this is possible by composing
$\exp(\exp(g,Y),X)$ since this is equivalent to the goal term by the
algebraic properties of exponentiation. It turns out that one can
define such semantics in a general, uniform, and concise way for an
\emph{arbitrary} algebraic theory as an intruder deduction
problem~\cite{Modersheim09,DBLP:journals/ipl/ChevalierR10,sps,DBLP:conf/birthday/BasinKRS15}.

As a side effect of ``abusing'' the intruder deduction 
to define the behavior of honest agents, one
prevents many specification errors that can easily happen in the
lower-level formalisms, e.g., when the message sent by one agent are
different from the messages that another agent expects, rendering the
protocol unexecutable. This may lead in the worst case to a false
negative (an attack of the real system is not detected because of a
specification error). Many such errors are prevented by formal Alice-and-Bob approaches because the protocol would be refused as
unexecutable by the compiler.

Thus, Alice-and-Bob notation is a beneficial and accessible 
specification language that can be used even without a deep 
background in formal verification. It is also striking how often
scientific works that formalize a protocol using a lower-level
language  first summarize the protocol informally in Alice-and-Bob
notation. 
However, existing formal Alice-and-Bob languages do not support branches in the protocol execution (by conditions or nondeterminism), unbounded repetition, or mutable long-term memory: everything is restricted to a linear session of fixed length.

While informal uses of Alice-and-Bob notation can easily be extended 
ad hoc, the first main contribution of this paper is a new
choreography language, called \textit{CryptoChoreo}, which extends
Alice-and-Bob notation with nondeterministic choice,  branching, and mutable long-term memory.\footnote{Since a choreography can be
  executed an unbounded number of times, nondeterminism and long-term
  memory are sufficient to formulate repetitions and sequential
  composition of protocols without an explicit repetition
  construct.}
  
Such features are needed for instance if we want to formulate a
protocol with a server that maintains a long-term database and that
may react to a request in different ways depending on the current state
of its database. Nondeterministic choice can be helpful for modelling
several options in a protocol that are at a participant's discretion,
where we do not want to formalize how they make a decision. Especially
this allows us to formulate an API where a user can
nondeterministically choose to send any of a number of commands to a
server (who may in turn ask other servers in order to answer the
request).

CryptoChoreo is, in a sense, a conservative extension of formal
Alice-and-Bob languages:  we give a semantics---parameterized
over an arbitrary algebraic theory---that agrees with standard
Alice-and-Bob languages on the subset that does not use the new
constructs. A particular challenge for this semantics is
to integrate the algebraic understanding of the protocol  with
 branching: if one party has made a nondeterministic choice, 
other parties do not necessarily know which choice was taken. For instance, our
semantics allows for the following protocol: Alice nondeterministically
chooses one of two types of message that she sends encrypted for Bob
over an intermediary server as an authentication service; the server
is unaware in which branch the execution is, but in each case it can execute its step uniformly by checking a MAC from Alice and signing the resulting message.

The semantics is formulated as a translation to local behaviors, i.e.,
a process for each role of the protocol. In general, this involves
algebraic problems that are not recursively computable (since,
e.g., whether two terms are equal under a set of algebraic equations is
in general undecidable). The second main contribution is to give a
computable translation for an algebraic theory that includes standard
constructors and destructors as well as exponentiation (for
Diffie-Hellman).

As a third contribution, we connect our translator with ProVerif and
demonstrate the effectiveness of our approach with several case
studies. A particular challenge is that ProVerif's abstraction often
is not precise enough when the long-term memory induces non-monotonic
behavior (e.g., when a certain action is possible only until a change
of the memory state) and thus fails to verify a protocol. We have
developed several heuristics to make sound encodings in ProVerif that
often overcome these problems.

Note that both the translation to process calculus and the target of
the translation, ProVerif, are using Dolev-Yao models, i.e., treating
cryptographic operations as blackboxes. The translation formalizes how
each role is \emph{supposed to execute} the choreography, and this
should not require the role to break the cryptography. Thus, it makes
sense that our semantics basically requires that a Dolev-Yao intruder
with the specified initial knowledge is able to correctly execute a
given role. For the target of the translation one may, however,
also consider computational verification like CryptoVerif. This would
require specifying more details like key-sizes and required
cryptographic properties of the encryption primitives, as well as
different kinds of goal specifications. Similarly, the translation
could also be used to generate a secure-by design implementation, in
the sense that the implementation creates outgoing, and checks
incoming, messages in the same way as the formal ProVerif model.

The rest of this paper is organized as follows: in
Section~\ref{sec:choreo} we define the syntax of CryptoChoreo
 and give an example; in Section~\ref{sec:projection} we define the semantics for
an arbitrary algebraic theory; 
in Section~\ref{sec:mec}, we describe the practical implementation for a
representative algebraic theory; in Section~\ref{sec:case:studies:meta} we describe the
connection to ProVerif and our case studies;
\added{we present related work in Section~\ref{sec:related:work};}
and we conclude in Section~\ref{sec:conclusions}.

\confver{
  \added{ An extended version of this paper, containing in particular
    the proofs omitted from Section~\ref{sec:mec} and details on the
    connection to ProVerif, is available at~\placeholder.}
}

\section{Choreography Language}\label{sec:choreo}
\begin{figure}[!t]
  \centering
  \begin{displaymath}
    \begin{array}{l@{\;\!}l}
      \!\role A : & (\nu M.\,\\
                  &\role A \rightarrow \role s: \scrypt((\role B,{\color{blue}\crypt((\msg,M),\ek(\role
                    B))}),\shk(\role A,\role s)).\\
                  & \role s \rightarrow \role B: \sign((\role A,{\color{blue}\crypt((\msg,M),\ek(\role
                    B))}),\inv(\sk(\role s))).\\
                  &\niauth{\role B}{\role A}{M})\\
      + & (\nu K.\,\\
                  & \role A \rightarrow \role s: \scrypt((\role B,{\color{blue}\crypt((\upd,K),\ek(\role
                    B))}),\shk(\role A,\role s)).\\
                  & \role s \rightarrow \role B: \sign((\role A,{\color{blue}\crypt((\upd,K),\ek(\role
                    B))}),\inv(\sk(\role s))).\\
                  &\secret{K}{\role A,\role B})
    \end{array}
  \end{displaymath}
  \caption{Example Choreography.}\label{fig:ex:choreo}
\end{figure}
Let us start with the example choreography in
Fig.~\ref{fig:ex:choreo}. We discuss later the front-matter
declarations (e.g., types and initial knowledge) that is needed for a
full specification.
Role $\A$ is starting the choreography and first makes a
nondeterministic choice ($+$) about which of the two
sub-choreographies to run. In the first case, $\role A$ generates a
new random $M$ (representing a message to send to $\role B$)\replaced{, pairs it}{ together}
with a constant tag $\msg$ (indicating that this is a message
transmission), then asymmetrically encrypts it with the encryption key
$\ek(\role B)$, highlighted in blue. $\role A$ then symmetrically
encrypts the blue message and the name of $\role B$ with a shared key
$\shk(\role A,\role s)$ with the (trusted) server $\role s$. Suppose
only $\role B$ knows the decryption key $\inv(\ek(\role B))$, then
$\role s$ can only decrypt the outer symmetric encryption, but not the
blue message. The next step is that $\role s$ signs the blue message
and the sender name $\role A$ using the private signature key
$\inv(\sk(\role s))$; the idea is that $\role s$ is vouching that the
blue message really came from $\role A$ (as the symmetric encryption
guarantees to $\role s$) and $\role B$ can verify the signature
knowing the corresponding public signing key $\sk(\role s)$. Finally,
we have an authentication goal when $\role B$ receives this message:
namely, that $\role A$ has indeed intended to send the message $M$ to
$\role B$. Non-injective ($\mathsf{ni}$ in $\mathsf{ni\text{-}authenticates}$) here means
that $\role B$ has no freshness guarantee (the message may be a
replay). The other sub-choreography is very similar, except that here
$\role A$ sends a different kind of message, a key update with a fresh
key $K$ with a different goal: $K$ is secret between $\role A$ and
$\role B$ (of course both sub-choreographies could have a secrecy and an
authentication goal).

Each of the two sub-choreographies could be specified in existing
formal Alice-and-Bob languages, but not the nondeterministic
choice. Note that $\role s$ here does not realistically know which
sub-choreography was chosen by $\role A$. The semantics we give below
sorts this out correctly: the server shall do the same operations in
both choreographies and simply handle the blue message as a black
box. Note that an intruder playing role $\role A$ is also allowed; and
this intruder may form a blue message that complies with neither
sub-choreography; $\role s$ will anyway accept this message if
everything it can check complies with the protocol.

\paragraph{Terms} We build terms using an alphabet $\Sigma$ of
function symbols and a set of variables $\mathcal{V}$. We denote all
function symbols with lower-case letters and all variables with
upper-case letters. In the above example, e.g., $\scrypt$, $\shk$,
$\inv$, $\msg$, and $\role s$ are function symbols (constants are
function symbols with $0$ arguments), while $\role A$ and $M$ are
variables. Variables mean that they can take a different value in
every run of the choreography. We use sans-serif font to denote
\emph{roles} of the protocol; they can be variables like $\role A$ and
$\role B$ or constants like $\role s$. The latter means that there is
one fixed player who cannot be the intruder---an easy way to specify a
trusted third party. We will discuss below the specification of
function symbols and their algebraic properties. We also use the
notation $(t_1,\ldots,t_n)$ for a concatenation using a pair
operator.

\begin{figure}[!t]
  \centering
  \begin{displaymath}\small
    \begin{array}{rcl@{\;\;}l}
      \C &::=& \inact & \text{(end)}\\
         &\mid& \role A \rightarrow \role B: t .\, {\C} & \text{(interaction)}\\
         &\mid& \secret{t}{\role A_1,\ldots,\role A_n}.\C & \text{(secrecy goal)}\\
         &\mid& \auth{\role B}{\role A}{t}.\C & \text{(inj. auth. goal)}\\
         &\mid& \niauth{\role B}{\role A}{t}.\C & \text{(non-inj. auth. goal)}\\ 
         &\mid& \role A : \A & \text{(atomic)}
      \\\\
      \A &::=&  \nu N.\; \A & \text{(new)}\\
         &\mid& \A_1+\A_2 & \text{(choice) }\\
         &\mid& \ite {s\doteq t}{\A_1}{\A_2} & \text{(condition)}\\ 
         &\mid& s:=c[t].\A & \text{(memory read)}\\
         &\mid& c[t]:=s.\A & \text{(memory write)} \\
         &\mid& \C & \text{(choreography)}
    \end{array}
  \end{displaymath}
  \caption{Syntax of CryptoChoreo.}
  \label{fig:syntax}
\end{figure}

\paragraph{Syntax} The formal syntax of a choreography is shown in
Fig.~\ref{fig:syntax}. A choreography $\inact$ represents a terminated
protocol, in which each participant has terminated. We omit trailing
$\inact$s when this is clear from the context.
An interaction $\role A \rightarrow \role B: t .\, {\C}$ denotes a
protocol where $\role A$ sends a term $t$ to $\role B$ and then
continues with choreography $\mathcal{C}$. The next items represent
the specification of goals (injective and non-injective
authentication, and secrecy) that we discuss later in detail.

All these constructs are present in existing formal Alice-and-Bob
languages. What CryptoChoreo is adding are constructs that are all
local to one role. Denote this by giving a role name $\role A$
followed by a colon and an \emph{atomic section} $\A$ of steps that
$\role A$ locally executes. Here, we have the fresh generation of a
random value 
$\nu N$ (as is standard).  
Next,
we have the nondeterministic choice $\A_1+\A_2$ (this is actually an
internal choice for the role who runs this atomic section and an
external choice for all others). Then, we have a conditional where the
condition is a comparison of terms. Last, we have reading from, and
writing to, long-term memory. We denote with $c[t]$ a memory cell in a family of memory cells $c$, where $c$ is an identifier and $t$ is an index
term. 
  On memory read and write, there will be no race conditions
with other parallel sessions, because our semantics will treat each
section $\A$ as atomic like the name suggests.

\paragraph{Memory Cell Example} To illustrate memory, consider the
following augmentation shown in \added{Fig.}~\ref{fig:ex:ext} of the example in
Fig.~\ref{fig:ex:choreo}, where $\role A$ may not know the encryption
key of every $\role B$ (but just the public signature verification key
$\sk(\role s)$ of $\role s$). When $\role A$ wants to talk to
$\role B$, she checks her memory cell $\mathit{keys}[\role B]$; if
this returns the initial value $\blank$, then she does not know the
key of $\role B$ and asks $\role s$, which we assume knows all public
encryption keys via the $\ek$ function and can vouch for it with its
signature. $\role A$ checks the signature and stores the key (note
that for $\role A$ the term $\ek(\role B)$ is just a blob that she
cannot verify in itself; this will be shown in the semantics below).
\begin{figure}[!t]
  \centering
  \begin{displaymath}
    \begin{array}{ll}
      \role A: & \mathit{EKB} := \mathit{keys}[\role B].\\
               & \mathsf{if}~\mathit{EKB}\doteq\blank~\mathsf{then}\\
               & \quad \nu N.\,\role A\rightarrow \role s:
                 (\mathit{key},\role B,N)\\
               & \quad \role s\rightarrow \role A:
                 \sign((\mathit{key},\role B,\ek(\role B),N),\inv(sk(\role s)))\\
               & \quad A: \mathit{keys}[\role B]:=\ek(\role B)\\
               & \mathsf{else} \ldots \text{(previous example with
                 $\ek(\role B)\mapsto {\mathit{EKB}}$)}
    \end{array}
  \end{displaymath}
  \caption[fig:ex:ext]{Extension of the example from
    Fig.~\ref{fig:ex:choreo}.}
  \label{fig:ex:ext}
\end{figure}
Note that we do not have a repetition operator, and this example shows
why this is without loss of generality: in the case that $\role A$
does not know the key of $\role B$, the run of the choreography ends
with writing the key (that she received from $\role s$) into her
memory. Thus, in any later run between the same $\role A$ and
$\role B$, $\role A$ will retrieve the key from her memory and run the
standard protocol with it. In other words, unbounded repetition is
implicit, because a choreography can be executed any number of times
(also in parallel) with arbitrary instances of the (non-constant)
roles and information between different runs can be transferred using
memory.

\paragraph{Front Matter}

The definition of our choreography language is parameterized over
sets $\Sigma$ and $\mathcal V$ (respectively, function symbols and
variables)
and a set of equations $E$ (over $\Sigma,\mathcal{V}$-terms). The set
$E$ induces a congruence relation $=_E$ on terms. In the
implementation of our translation, in
Section~\ref{sec:mec}, 
we instantiate $E$ with a concrete choice of properties.  

Some variables and functions symbols are declared as roles (set in
sans-serif in this paper) and only these can be used in places where
the syntax indicates sans-serif font. For each role, one must declare
the \emph{initial knowledge}: a list of terms where all occurring
variables are of type role. For our example (in the augmented version
where $\role A$ does not initially know $\role B$'s public key), this
declaration could be:
\begin{displaymath}
  \begin{array}{ll}
    \mathsf{A}: & \mathsf{A,B,s,\sk(s),\shk(A,s)}\\
    \mathsf{s}: & \mathsf{A,B,s,\sk(s),\inv(\sk(s)),\shk(A,s),\ek(B)}\\
    \mathsf{B}: & \mathsf{A,B,s,\sk(s),\ek(B),\inv(\ek(B))}
  \end{array}
\end{displaymath}
Note that with less knowledge the protocol would not be executable
(neither would the initial version of the example where $\role A$ cannot
ask the server for $\role B$'s public key be). 

We require that for every variable that is \emph{not} of type role,
the first occurrence is either in a new statement (like $\nu X$) or in
a memory read (like $X:=c[s]$). Also, in a new statement $\nu X$, we
assume that $X$ did not occur before in the choreography (this can be
achieved by renaming). In contrast, variables in a memory read may
have occurred before, e.g., $\nu X.\ldots X:=c[s_1]. \ldots X:=c[s_2]$
is legal, and it would mean that the value retrieved here is the same
value as before---in an ideal unattacked execution of the
choreography. Our semantics can tell if the respective role has the
necessary knowledge to check that and insert such a check in the code,
if so.


\section{Translation Semantics}\label{sec:projection}

The semantics of CryptoChoreo is now defined by a translation to a
process calculus, where we define for each role of the choreography a
process, representing the local behavior of this role in one execution
of the choreography. We then allow arbitrary instances of all roles to
run in parallel, together with an intruder who can also play any of the
roles (except trusted third parties) as a normal participant (but who
does not necessarily follow the protocol).

\subsection{Local Behaviors}\label{sec:local}
\begin{figure}[t]
  \centering
  \begin{displaymath}\small
    \begin{array}{rcl@{\qquad}l}
      \LC &::=& \inact & \text{(end)}\\
          &\mid& \send(r).\LC_i & \text{(send)}\\
          &\mid& \receive(\X).\LC & \text{(receive)}\\
          &\mid& \lock.\A & \text{(atomic local)}
      \\\\
      \A  &::=&  \nu \X.\, \A & \text{(new)}\\
          &\mid& \A_1 + \A_2 & \text{choice}\\
          &\mid& \ite {r_1\doteq r_2}{\A_1}{\A_2} & \text{(cond)}\\
          &\mid& \X:=c[r_1].\A & \text{(mem read)}\\
          &\mid& c[r_1]:=r_2.\A & \text{(mem write)}\\
          &\mid& \event(r).\A & \text{(event)}\\
          &\mid& \unlock.\LC & \text{(end)}
    \end{array}
  \end{displaymath}
  \caption{Syntax of Local Behaviors.}
  \label{fig:syntax:local}
\end{figure}

It is convenient for the translation and the later connection to
ProVerif to define a restricted syntax and semantics for local
behaviors as the target language of the translation semantics.

\paragraph{Syntax}
The syntax of local behaviors mirrors that of CryptoChoreo from the
point of view of a single role: instead of a communication step
between two roles, we have now sending and receiving. We use here
however a different set of symbols $\Sigma_p$ which represent public
functions, i.e., functions that every agent, including the intruder,
can apply. This will include most of the functions from $\Sigma$ like
$\crypt$ that represent cryptographic algorithms, as well as public
constants like $\msg$ in the example above. It will not include,
however, some functions from $\Sigma$ that just describe relations in
the model, but do not represent actual cryptographic algorithms like
$\inv$ (which maps public keys to the corresponding private key) or
$\ek$ (which maps an agent to a corresponding public key). $\Sigma_p$
will also include some functions that are not in $\Sigma$: observe
that we have not used any functions for decryption or signature
verification, because they are \emph{destructors} or \emph{verifiers},
i.e., functions that extract a subterm or verify the structure of a
term; while the choreography is only concerned with \emph{constructed}
messages. Also, we will use a distinct set of variables called
\emph{labels}. Labels are denoted $\X_1,\X_2,\ldots$ and are disjoint
from $\Sigma$, $\Sigma_p$, and $\mathcal{V}$. The terms built from
these variables and $\Sigma_p$ are called \emph{recipes} and we denote
them with $r$, $r_1$, $r_2$, etc.

Fig.~\ref{fig:syntax:local} shows the syntax of local behaviors. Note
that, in order to enter the atomic section, it is necessary to make a
$\lock$ step and it has to end with an $\unlock$ step; the semantics
of local behaviors use that as a mutual exclusion mechanism on the
memory to prevent race conditions (even if one ensures by design that
each memory cell belongs to a particular agent, there may be more than
one run of the choreography in parallel). 
The other constructs are similar to choreographies. However, instead of
variables like $M$ we have now labels like $\X$, and instead of terms
(over $\Sigma$ and $\mathcal{V}$) we have recipes. Note that memory
read and receiving can only ``read into'' a label $\X$ (while on the
choreography level, these can be composed terms). We require that at
each receive and memory read, we use a new label (that did not occur
before in the local behavior; this can be easily achieved by
renaming).

\paragraph{Frames} To capture the knowledge of an honest agent at a
state of protocol execution, we define a \emph{frame} to be a finite
mapping $F=[\X_1\mapsto t_1,\ldots,\X_n\mapsto t_n]$ where the $X_i$
are labels and the $t_i$ are terms. We call $\{\X_1,\ldots,\X_n\}$ the
\emph{domain} of $F$ and we say $F$ is \emph{concrete} if the $t_i$
contain no variables. The semantics of local behaviors will only use
concrete frames. 
Given a recipe $r$, we use a frame $F$ like a
substitution and write $F(r)$ for the term that results from replacing
the labels $\X_i$ with the respective term $t_i$; $F(r)$ is undefined
if $r$ contains labels outside the domain of $F$.

To each role $\role A$, we attach an initial knowledge frame
$F_{\role A}$ that is not necessarily concrete but contains only
variables of type role. In the translation from CryptoChoreo to local
behaviors, we take the initial knowledge of each role
$\role A: t_1,\ldots,t_n$ at the CryptoChoreo level and turn into an
initial knowledge frame $F_{\role A}=[\X_1\mapsto
t_1,\ldots,\X_n\mapsto t_n]$ for the local behavior.

We require that all labels in a local behavior first occur in the
initial knowledge frame, in a new, in a receive, or in a memory read.

\paragraph{Example Local Behavior} The role $\role A$ of the example
of Fig.~\ref{fig:ex:ext} will be translated by the semantics into the
following local behavior:
\begin{displaymath}{
\hspace*{-3mm}  \begin{array}{l}
F_{\role A}=[\X_1\mapsto \role A, \X_2\mapsto \role B,
\X_3\mapsto \role s, \X_4\mapsto \sk(\role s), \X_5\mapsto \shk(\role
    A,\role s)]\\
    \lock.\X_{K}:=\mathit{keys}[\X_2].\\
    \mathsf{if}~\X_K\doteq \blank~\mathsf{then}
    \\
    \quad \nu
    \X_N. \unlock.\send((\mathit{key},\X_2,\X_N)).\\
    \quad\receive(\X_{kc}).\lock.\\
    \quad \mathsf{if}~(\vsign(\X_{kc},\X_4)=\top)~\mathsf{then}\\
    \qquad \mathsf{let}~(\X_t,\X_B,\X_{K'},\X_{N'})=\open(\X_{kc})~\mathsf{in}\\
    \qquad \mathsf{if}~(\X_t\doteq \mathit{key} \wedge \X_B=\X_2 \wedge
    \X_{N'}=\X_N)~\mathsf{then}\\
    \qquad \quad \mathit{keys}[\X_2]:=\X_{K'}.\unlock.0\\
    \qquad \mathsf{else}~\unlock.0\\
    \quad \mathsf{else}~\unlock.0\\
    \mathsf{else}~...~\text{translation of the other part, using
    $\X_K$ as $\ek(\role B)$}
  \end{array}}
\end{displaymath}
Here, we use some syntactic sugar: several checks can be done by one
condition; and that we write let for parsing the content of a message,
in this case expecting that it can be parsed into a
quadruple. The function
$\vsign$ is supposed to be signature verification with the property
$\vsign(\sign(m,\inv(k)),k)=_E\top$ and open a destructor that yields
the signed message, i.e., $\open(\sign(m,\inv(k)))=m$. This models a
signature scheme where the signed text is transmitted in plain along
with a signed hash, i.e., one needs the public key only for signature
verification. One can also observe that between each $\lock$ and
$\unlock$ there is at most one memory read or write, so locking is in
this case actually redundant as it does not prevent any race
conditions.

\begin{figure}[t]
  \centering
  \begin{displaymath}
    \begin{array}{l}
      (0 \uplus L, F,\mu) \to (L, F,\mu)
      \\
      (\send(r).\LC \uplus L, F,\mu)
      \to (\LC \uplus L, F[\X \mapsto r],\mu )\\
      \quad\text{ where $\X$ is a fresh label\added{ with regard to $F$}}
      \\
      (\receive(\X).\LC \uplus L, F,\mu)
      \to (\LC[\X \mapsto F(r)]\replaced{\uplus}{\cup} L, F,\mu)\\
      \quad\text{ where $r$ is a recipe over the domain of }F
      \\
      (\lock.\A\uplus L,F,\mu) \stackrel{\mathit{tr}}{\to} 
      (\LC\uplus L,F,\mu') \\
      \quad\text{ if } (\A,\mu)\stackrel{\mathit{tr}}{\Rightarrow^*}(\unlock.\LC,\mu')  \\
      (L, F,\mu) \to (\sigma(F_{\role A})(\LC_{\role A}) \uplus L, F,\mu)\\
      \quad\text{if $\sigma$ is \deleted{a} an honest instantiation of a role
      $F_{\role A}:\LC_{\role A}$}\\
      (L, F,\mu) \to (L, F\uplus \sigma(F_{\role A}),\mu)\\
      \quad\text{if $\sigma$ is a dishonest instantiation of a role
      $F_{\role A}:\LC_{\role A}$}\\
      \quad\text{and where the labels of $F_{\role A}$ have been freshly renamed}
      \\\\
      (\nu\X.\, \A, \mu)
      \Rightarrow (\A[\X \mapsto n],\mu)
      \text{ where $n$ is a fresh constant}
      \\
      (\A_1+\A_2,\mu) \Rightarrow (\A_i,\mu) \text{ if }i\in\{1,2\}
      \\
      (\ite {r_1\doteq r_2} {\A_1} {\A_2}, \mu) \Rightarrow (A_1,\mu) \text{ if
      } r_1=_E r_2
      \\
      (\ite {r_1\doteq r_2} {\A_1} {\A_2}, \mu) \Rightarrow (A_2,\mu) \text{ if
      } r_1\neq_E r_2
      \\
      (\X:=c[r_1].\A,\mu) \Rightarrow (\A[\X\mapsto \mu(c,r_1)],\mu)
      \\
      (c[r_1]:=r_2.\A,\mu) \Rightarrow (\A,\mu[(c,r_1)\mapsto r_2])
      \\
      (\event(r).\A,\mu) \stackrel{r}\Rightarrow (\A,\mu)
    \end{array}
  \end{displaymath}
  
  \caption{Semantics of Local Behaviors.}
  \label{fig:semantics:local}
\end{figure}
\paragraph{Semantics of Local Behaviors}

We give a simple operational semantics for a set of local behaviors
$F_{\role A_1}:\LC_{\role A_1},\ldots,F_{\role A_n}:\LC_{\role A_n}$
(where each $F_{\role A_i}$ is the initial knowledge frame of
$\role A_i$).
We assume a set $\Ag\subseteq \Sigma\cap\Sigma_p$ of public constant
of type role and that $\i\in\Ag$ represents a dishonest agent
(``intruder'') while all other agents are honest.\footnote{One may
  well consider more than one dishonest agent, but for simplicity we
  work with just one.}

For a behavior $F_{\role A}:\LC_{\role A}$ we call the substitution
$\sigma$ an \emph{instantiation} if it maps all variables in
$F_{\role A}$ (that are by definition of type role) to elements of
$\Ag$. We say it is a \emph{dishonest instantiation} if
$\sigma(\role A)=\i$ and an \emph{honest instantiation} otherwise. We
write $\sigma(F_{\role A})$ for the instantiation of the initial
knowledge frame and $\sigma(F_{\role A}\added{)}(\LC_{\role A})$ for the
instantiation of the behavior itself, replacing all labels from
$F_{\role A}$ in $\LC_{\role A}$ by ground terms; thus all remaining
labels first occur at a new, at a receive, or at a memory read.

The last ingredient for the semantics is a memory map $\mu$ that maps
every memory cell $c[(t)_E]$ \replaced{to}{gives} a value, initially $\blank$, where
$(t)_E$ is the $=_E$-equivalence class of a ground term $t$ over
$\Sigma$. As easy notation we just write $\mu(c,t)$ for this value,
and we write $\mu[(c,t)\mapsto t']$ for changing the memory cell
$c[(t)_E]$ to value $t'$.

The semantics of local behavior is defined in
Fig.~\ref{fig:semantics:local} and consists of two transition
relations $\to$ and $\Rightarrow$ that call each other: $\to$ is on
triples $(L,F,\mu)$ where $L$ is a multi-set of local behaviors, $F$
is a frame representing the intruder knowledge and $\mu$ is the memory
map; the initial state is $(\emptyset,[],\mu_0)$ where $\mu_0$ maps
everything to $\blank$; $\Rightarrow$ is on tuples $(\A,\mu)$ where
$\A$ is an atomic section of a local behavior. 
We decorate the transition relations
with a list of events that occurred upon the transitions. In this
semantics, the intruder \emph{is} the network: every message an honest
agent sends gets added into the intruder knowledge, and every message
an honest agent receives comes from the intruder knowledge: the
intruder can choose any recipe over their knowledge, which includes
encrypting and decrypting with known keys. An atomic section is
handled literally atomically using the $\Rightarrow$ relation locally
at an agent until it hits the $\unlock$; we label the transition with
the trace $\mathit{tr}$ of all events that the agent emitted.
\replaced{
The last two rules regarding $\to$ allow spawning}{The next
rule allows to spawn} new instances of any role $\role A$: we choose any
instantiation $\sigma$ of the variables in $F_{\role A}$ with agent
names; if $\sigma$ is honest, i.e., $\sigma(\role A)\neq \replaced \i i$, we apply
the instantiated knowledge $\sigma(F_{\role A})$ as a substitution to
the local behavior $\LC_{\role A}$, leaving only labels that are
introduced by new, receive, and memory read. This semantics
allows running an arbitrary number of sessions in parallel and
sequentially.  If $\sigma$ is dishonest, i.e., $\sigma(\role A)=\replaced \i i$
then this represents that the intruder plays role $\role A$ under the
actual name $\replaced \i i$. This models a dishonest/compromised agent. We give
the intruder the initial knowledge needed to play the role, i.e.,
$\sigma(\F_{\role A})$ where we have to rename the labels in the frame
to avoid a clash with labels in the present intruder
knowledge.


\subsection{Projection: The Semantics of CryptoChoreo}
We can now give the semantics of CryptoChoreo by translation to local
behaviors. We again use frames to represent the knowledge of a role at
a given point in the translation, but this differs from their use in
the local behavior semantics. As said, the messages in a choreography
represent how messages look like in an ``ideal'' or unattacked
run---which may differ from the shape of messages in a real run due to
interference from the attacker. For instance if $\role A$ is supposed
to receive $\exp(g,Y)$ from $\role B$ for a secret $Y$ that $\role A$
does not know, there is nothing that $\role A$ can check about this
message. To keep track of this during translation we make an entry
$[\X_i\mapsto \exp(g,Y)]$ in the frame $F$ of $\role A$ that
expresses: $\role A$ has received \emph{some} message $\X_i$ and
according to the choreography it is \emph{supposed to be}
$\exp(g,Y)$. Given that another entry $[\X_j\mapsto X]$ represents the
fresh value $\role A$ has created for her own secret $X$, then the
Diffie-Hellman key $\exp(\exp(g,X),Y)$ can be formed with the recipe
$r=\exp(\X_i,\X_j)$: $F(r)=\exp(\exp(g,Y),X)=_E \exp(\exp(g,X),Y)$. In
this way, frames make the connection between the messages the agents
have and use in their local behavior and what the messages supposedly
are on the choreography level. \deleted{We will use this to figure out
  how agents generate outgoing messages and how they analyze incoming
  messages. We thus define two core algorithmic problems on frames:}
\added{As part of the translation semantics, we thus need to give a
  general definition of how an agent should generate an outgoing message
  from a given frame and this is the first core algorithmic problem we
  define on frames:}
\begin{itemize}
\item The deduction problem: Given frame $F$ and term $t$, compute a
  recipe $r$ such that $F(r)=_E t$ if one exists or return fail
  otherwise.
\end{itemize}
\added{Here, the failure means that there is no feasible way for the
  agent to create the required outgoing message. Our semantics will
  in this case refuse the choreography as unexecutable. This will
  typically happen in case of a specification error, for instance if an agent
  does not possess the necessary keys to participate in the protocol as
  suggested.}

\added{For receiving messages there is a related problem to
  solve. Suppose the choreography contains a step where an agent
  should receive a hash $h(N)$ of a random number $N$. Since $h$ is
  one-way and $N$ random, the agent cannot perform any checks on this
  message, so will have a frame entry $[\X_i\mapsto h(N)]$. Suppose in
  a later step, the agent is supposed to receive $N$, i.e.,
  $[\X_j\mapsto N]$. Now there is something the agent can check:
  whether $F(\X_i)=_EF(h(\X_j))$ holds. In general, a check is a pair
  of recipes (here $\X_i$ and $h(\X_j)$) that are supposed to give
  the same message under the present frame $F$. The agent must not
  proceed with the execution if this check fails, because at least one
  of the incoming messages does not comply with the protocol. In a
  similar way, if an agent receives a message that is supposed to be
  an encryption with a key for which the agent has the decryption key,
  then there must be a check that this decryption indeed works. It is
  one of the main features of CryptoChoreo that the semantics defines
  what checks each party must do on incoming messages, while the
  modeler only specifies in the choreography what the messages look
  like in an unattacked protocol run.}

\added{ However, in general there are infinitely many possible checks,
  for instance in the discussed situation one could also check
  $F(h(\X_i))=_EF(h(h(\X_j)))$, but this check seems redundant given
  the other check above. To make this notion of redundant checks
  precise, we adapt some equational logic concepts. Define an
  \emph{interpretation $\I$} as a mapping from all labels to ground
  terms. (A frame is similar to an interpretation but it has only a
  finite domain.) For a recipe $r$, let $\I(r)$ be the ground term
  that results from replacing every label $\X$ in $r$ by $\I(\X)$. We
  now say that $\I$ is an \emph{$E$-model} of the formula
  $r_1\doteq r_2$, and write $\I\models_E r_1\doteq r_2$, if
  $\I(r_1)=_E\I(r_2)$, and we extend this to conjunctions of equations
  as expected. We finally say that \emph{$\phi$ $E$-implies $\psi$},
  and write $\phi\models_E \psi$, if every $E$-model of $\phi$ is
  also an $E$-model of $\psi$. For instance $\X_i\doteq
  h(X_j)\models_E h(\X_i)\doteq h(h(\X_j))$.}

\added{
  We can now define the second core algorithmic problem: to find a
  finite set of checks $\phi$ that is complete in the sense that any
  other checks one could make are already implied by $\phi$ already:
}
\begin{itemize}
\item The complete check problem: Given a frame $F$, compute a finite
  set of \emph{checks}, i.e., equations of recipes
  $\phi=\{r_1\doteq r_1',\ldots,r_n\doteq r_n'\}$ such that
  $F(r_i)=_E F(r_i')$ for each $1\leq i\leq n$, that is
  \emph{complete} in the sense that if for any other $r_0, r_0'$ we have $F(r_0)=_E F(r_0')$, then
  $\phi\models_E r_0\doteq r_0'$, or return fail if no finite set of
  checks satisfies that.
\end{itemize}
\added{A failure case, i.e., $E$ and $F$ where every complete set of
  checks is infinite, is surprisingly hard to construct, see
  Appendix~\ref{sec:infinite:checks}, so we do not consider this a problem in
  practice.}
\added{In section~\ref{sec:mec} we give an example of an algebraic
  theory $E$ for which both the deduction problem and the finite
  complete set of checks problem are decidable, and we sketch how
  these algorithms work. \confver{The details and the proofs are given in the
  extended version. }\extver{The details and the proofs are given in Appendix~\ref{app:mech}. }}
  \added{Note that for this
  $E$, every frame has a finite complete set of checks; while the
  construction of recipes can of course still fail when there is
  insufficient knowledge in a frame to compose the outgoing
  message.}

\replaced{For arbitrary $E$, both}{Both} these problems are in general
not recursively computable (because in general even $=_E$ is
undecidable).\deleted{but in Section~\ref{sec:mec} we give algorithms
  for both problems for a standard set $E$ of equations.}
  \added{Our projection semantics is parameterized by an arbitrary equational theory, while in Section~\ref{sec:mec} we provide procedures for one representative equational theory.}

We first note a complete set of checks for $F$ is not unique, however
if $\phi$ and $\psi$ are two complete set of checks for $F$, then
$\phi\models_E\psi$ and $\psi\models_E\phi$, so they are equivalent
and in the semantics we can leave this choice
undetermined\footnote{Thus a concrete implementation is free to choose
  one.} without making the semantics ambiguous. By abuse of notation
we thus write $\phi(F)$ for a complete set of checks for $F$, even
though it is, strictly speaking, not a function.

Second, also the deduction problem has in general many solutions,
i.e., different $r_1$ and $r_2$ such that $F(r_1)=_EF(r_2)=_E
t$. \added{For instance, if $F=[\X_i\mapsto h(N), \X_j\mapsto N]$, and
  the agent should send $t=h(N)$, then there are two recipes to
  construct $t$: $r_1=\X_i$ and $r_2=h(\X_j)$. This choice of recipes
  could make a difference if the agent received for $\X_i$ or $\X_j$
  terms that do not comply with the protocol. However, after the agent
  has successfully executed the checks
  $\phi(F)=\{\X_i\doteq h(\X_j)\}$, both recipes are guaranteed to
  produce the same term.}  \replaced{More generally, if
  $F(r_1)=_E F(r_2)=_E t$, then $\phi(F)\models_E r_1\doteq r_2$,}{In
  this case $\phi(F)\models_E r_1\doteq r_2$,} i.e., if we have
performed all the checks in $\phi(F)$, then also the choice between
the two recipes $r_1$ and $r_2$ does not matter.

Third, in the semantics, we need a slight generalization of the
deduction problems, namely given several frames $F_1,\ldots,F_n$ and
goal terms $t_1,\ldots,t_n$ and we want a \emph{single} recipe $r$
that solves all \emph{deduction} problems, i.e., $F_i(r)=t_i$ for
every $1\leq i\leq n$. Suppose we already have a set $\phi$ of checks
that is a complete set of checks for each of the $F_i$, and suppose
there is a solution $r$ for all frames. Then any solution $r'$ for one
of the frames, say $F_1$, must be equivalent to $r$, i.e.,
$\phi\models_E r\doteq r'$. Thus for checked frames it suffices to
compute a solution for one frame and check if it works on the other
frames---if not, then there is no common solution for all frames.

Nondeterminism and conditions mean that, in general, a role does not
know which branch of the choreography we currently are in, and this
also holds during the translation. Therefore, during \deleted{in }the
translation, the translation state contains a finite set of pairs
$(F_i:\C_i)$ where each $F_i$ is a frame (all $F_i$ have the same
domain) and $\C_i$ is the remainder of the choreography that still
needs to be translated.
\added{
We note that this is to handle nondeterminism \emph{external} to
the role being translated; when an agent branches locally 
we handle the translation of each branch in isolation. 
}

\newcommand{\chor}{\mathsf{c}}
\newcommand{\atom}{\mathsf{a}}
\begin{definition}
  A \emph{translation state} is of the form
  $$(\role
  A,\phi\triangleright\psi,b,\{(F_1:\C_1),\ldots,(F_n:\C_n)\})$$
  where $\role A$ is the role we are currently translating; $\phi$ and
  $\psi$ are \replaced{sets}{a set} of equations $r_1\doteq r_2$ between recipes, where
  $\phi$ represents checks that have already been done, and $\psi$ are
  checks that are pending; $b\in\{\chor,\atom\}$ is a flag indicating whether
  we are on the choreography level or in an atomic section; the $F_i$
  are frames with the same domain that map to terms; and the $\C_i$
  are either choreographies if $b=\chor$ or atomic sections if
  $b=\atom$.

  During translation we preserve the \emph{invariant} that
  $\phi\cup\psi$ is \replaced{covering all checks that can be made in
    any frame, i.e., for every frame $F_i$ ($i\in\{1,\ldots,n\}$) and
    any pair of recipes $r_0,r_0'$ with $F_i(r_0)=F_i(r_0')$, it must
    hold that $\phi\cup\psi\models_E r_0\doteq r_0'$.}{a complete set
    of checks for the $F_1,\ldots,F_n$ where $\psi$ may contain checks
    that do not hold on all $F_i$.} If $\psi\neq 0$, i.e., if there
  are pending checks, they will be performed first before all other
  translation steps.

  Given a choreography $\C$ and a role $\role A$ of that
  choreography\replaced{, let $F_A=[\X_1\mapsto t_1,\ldots,\X_n\mapsto
    t_n]$ where $\{t_1,\ldots,t_n\}$ is the initial knowledge of role
    $\role A$ and the $\X_i$ are distinct labels. The}{and a frame $F_A$ that contains the initial knowledge of
  $\role A$ in the choreography specification where each item has
  received a unique label $\X_i$, the} \emph{initial translation
    state} for translating $\role A$ in $\C$ is:
  $$(\role A, \emptyset\triangleright\emptyset, \chor, \{(F_A:\C)\})$$
  i.e., there is just one possibility $\C$ where we are and the current
  knowledge is $F_A$.
\end{definition}

\subsubsection{Cases of the Semantics Function}
The semantics function $\sem{T}$ takes a translation state $T$ and
projects the choreography to the actions of the role, yielding a local
behavior for that role. We define it recursively by a case distinction
on $T$. Since we will often require all possibilities $(F_i,\C_i)$ to
start with the same kind of command, we use the following notation:
$\{F_i,\nu N_i.\C_i\}_{i=1}^{n}$ for
$\{F_1,\nu N_1.\C_1,\ldots,F_n,\nu N_n,\C_n\}$, and similar for other
constructs in place of $\nu N$. For simplicity, we first present this
semantics without goals.

\added{
  The semantic function translating a choreography to the local behavior of a given agent, $\role{A}$, 
  is defined by 14 cases, labelled a)-n).
  Case a) completes the projection when all the possible continuations are finished.
  Cases b)-d) handle communication steps: 
  If $\role{A}$ is supposed to send a message we must deduce a recipe to do so from the terms available in their frame. 
  If $\role{A}$ receives a message we must compute all the checks that may be performed on the received term and add 
  them to the set of pending checks. 
  Case e) handles the processing of pending checks. 
  Case f) handles the situation where the checks have ruled out any possible continuation. 
  Case g) handles the entering of an atomic section by the given agent,
  and case h) handles atomic sections of other agents. 
  Cases i)-m) handle the different actions that may be performed in an atomic section. 
  Each of these cases corresponds to a primitive action which is
  directly translated into local behavior.
  Finally, case n) states that the choreography is ill-specified if none of the previous cases apply.
}

\paragraph{{$\sem{\role A,\phi\triangleright\emptyset,\chor,\{(F_i:\inact)\}_{i=0}^n}$}  where $n>0$}\label{rule:finished}
All possibilities have finished, and the translation is
simply:\\ $\inact$.

\paragraph{{$\sem{\role A,\phi\triangleright\emptyset,\chor,\{(F:\role
      B\rightarrow \role C:t.\C)\}\cup\{F_i:\C_i\}_{i=1}^n}$} for $\role A\neq\role B$ and $\role A\neq \role C$}\label{rule:comm:other}
One of the
possibilities is a communication step that $\role A$ is not involved
in and is therefore ignored. The translation is thus:\\
$\sem{\role
  A,\phi\triangleright\emptyset,\chor,\{(F:\C)\}\cup\{F_i:\C_i\}_{i=1}^n}$

\paragraph{{$\sem{\role A,\phi\triangleright\emptyset,\chor,\{F_i:\role
    A \rightarrow \role B_i: t_i.\C_i\}_{i=1}^n}$} where $n>0$}\label{rule:send}
All possibilities are send steps for $\role A$.

As explained before, we check whether there is a recipe $r$ such
that $F_i(r)=t_i$ for each $1\leq i\leq n$.  If there is no such $r$,
then we reject the protocol as unexecutable: either there is no way
for $\role A$ to produce the outgoing term $t_i$, or the different
possibilities would require different recipes, and $\role A$ cannot
know in which possibility they are. However, if there is such an
$r$,\footnote{If there are several such recipes, the choice between
  them leads to equivalent translation outcomes as explained before.}
then the translation is:
\\
$\mathsf{send}(r). \sem{\role
  A,\phi\triangleright\emptyset,\chor,\{F_i:\C_i\}_{i=1}^n}$

\paragraph{$\sem{\role
  A,\phi\triangleright\emptyset,\chor,\{F_i:{\role B_i}\rightarrow
  \role A: t_i.\C_i\}_{i=1}^n}$ where $n>0$}\label{rule:receive}
All possibilities are receive steps for $\role A$.  

Let $\X$ be a new
recipe variable and $F_i'=F_i[\X\mapsto t_i]$ for every
$i\in\{1,\ldots,n\}$. Let $\phi_i$ be a complete finite set of checks
for $F_i'$ and let $\Phi=\cup_{i=1}^{n}\,\phi_i$. This represents all
checks that we can do in any of the frames $F_i$. First we can remove
from $\Phi$ all those checks that are already implied by the checks
$\phi$ from the translation state (i.e., that have already been done
before in the translation process). We can also remove from  $\Phi$ any
equation that is implied by the other equations. Let thus $\Phi'$ be
a resulting minimal set of equations.\footnote{Again, there may be
  several minimal sets, e.g., if two equations imply each other;
  however all resulting sets from the minimization are logically
  equivalent.} The translation of the receive step is then obtained by adding
received message to the frames and inserting the $\Phi'$ as pending
checks that have to be done next:\\
$\mathsf{receive}(\X).\sem{\role A,\phi\triangleright\Phi',\chor,\{F_i':\C_i\}_{i=1}^n}$

\paragraph{$\sem{\role A,\phi\triangleright\{r_1\doteq
    r_2\}\cup\psi,b,\{F_i:\C_i\}_{i=1}^n}$  where $n>0$}\label{rule:check}
There is at least
one pending check $r_1\doteq r_2$. 

We partition the possibilities
into those where $F_i$ satisfies the check and those that do
not:\\
Let $\mathsf{FCs_+}=\{(F_i:C_i)\mid F_i(r_1)=_E F_i(r_2)\}$\\ and
$\mathsf{FCs_-}=\{(F_i:C_i)\mid F_i(r_1)\neq_E F_i(r_2)\}$.\\
The translation is now:\\
$ \mathsf{if}\, r_1\doteq r_2\; \mathsf{then}\,
\sem{\role A,\phi\cup\{r_1\doteq r_2\}\triangleright\psi,b,\mathsf{FCs_+}}$\\
$ \phantom{\mathsf{if}\, r_1\doteq r_2\; } \mathsf{else}\,\sem{\role
  A,\phi\triangleright\psi,b,\mathsf{FCs_-}}$

\paragraph{$\sem{\_,\_\triangleright\_,\_,\emptyset}$ }\label{rule:nowai} 
There are no possible continuations.

This can happen
when doing a check that splits the possibilities into $FCs_+$
and $FCs_-$, and one of them is empty. 
It means that if we reach this branch, the
agent has detected that an incoming message is not compliant with the
choreography, and aborts the execution. The translation is thus
simply:\\
$\inact$

\paragraph{$\sem{\role
      A,\phi\triangleright\emptyset,\chor,\{F_i:\role A:\A_i\}_{i=1}^n}$ where $n>0$}\label{rule:atomic:me}
All possibilities start with an atomic section of the
  agent $\role A$.

  Then the translation is simply to issue the lock and switch the
  atomic section flag:\\
$ \lock.\sem{\role
      A,\phi\triangleright\emptyset,\atom,\{F_i:\A_i\}_{i=1}^n}$

\paragraph{{$\sem{\role A,\phi\triangleright\emptyset,\chor,\{(F:\role
      B:\A)\}\cup\{F_i:\C_i\}_{i=1}^n}$ } where $\role B\neq\role A$}\label{rule:atomic:other}
One possibility is that another role goes into its atomic
section. 

Role $\role A$ should ignore these steps and just extract all
continuations after the atomic section, which is defined as follows:
\begin{displaymath}
  \begin{array}{rcl}
    \cont(\A_1+\A_2)&=&\cont(\A_1)\cup\cont(\A_2)\\
    \cont(\ite{s\doteq t}{\A_1}{\A_2})&=&\cont(\A_1)\cup\cont(\A_2)\\
    \cont(\_.\A)&=&\cont(\A)\\
    \cont(\C)&=&\{\C\}
  \end{array}
\end{displaymath}
Let $\C_1,\ldots,\C_m=\cont(\A)$ in the translation:\\
$\sem{\role A,\phi\triangleright\emptyset,\chor,\{F:\C_i\}_{i=1}^m\cup\{F_i:\C_i\}_{i=1}^n}$

We now come to the cases for an atomic section of the agent we
translate for:
\paragraph{{$\sem{\role A,\phi\triangleright\emptyset,\atom,\{F_i:\nu
    N_i.\A_i\}_{i=1}^n}$} for $n>0$} \label{rule:nu} All possibilities create a
fresh number $N_i$. 

We pick a fresh label $\X$ and translate:\\
$\nu \X.\,\sem{\role
  A,\phi\triangleright\emptyset,\atom,\{F_i[\X\mapsto N_i]:\A_i\}_{i=1}^n}$

\paragraph{$\sem{\role A,\phi\triangleright\emptyset,\atom,\{F_i:
    s_i:=c[t_i].\A_i\}_{i=1}^n}$ where $n>0$} \label{rule:mem:read} 
All cases start with a memory
read. 

Similar to the send case, we require that
there is one recipe $r$ such that $F_i(r)=_Et_i$ for each
$1\leq i\leq n$. If not, the semantics rejects the protocol as
unexecutable at this point (because the agent either cannot create the
proper index for the memory lookup, or there are contradicting
possibilities for this index). 
The retrieved message $s_i$
is treated like in the receive case: we add it to the knowledge with a new
label $\X$, giving frames $F_i'=F_i[\X\mapsto s_i]$, and then we
compute a complete set of checks $\phi_i$ for each $F_i'$, compute the
union $\Phi=\cup_{i=1}^n\phi_i$, and remove redundant equations
leading to a reduced $\Phi'$. The translation is then:\\
$\X:=c[r].\sem{\role A,\phi\triangleright\Phi',\atom,\{F_i':\A_i\}_{i=1}^n}$

\paragraph{$\sem{\role A,\phi\triangleright\emptyset,\atom,\{F_i:
    c[t_i]:=s_i.\A_i\}_{i=1}^n}$ where $n>0$}\label{rule:mem:write}
All possibilities start with a write step.
We require that there are recipes $r_1$ and $r_2$ such that
$F_i(r_1)=t_i$ and $F_i(r_2)=s_i$ for each $i\in\{1,\ldots,n\}$. (If
not, this is a specification error, because it
is unclear what $A$'s next step is.) Then the translation is:\\
$c[r_1]:=r_2.\sem{\role
  A,\phi\triangleright\emptyset,\atom,\{F_i:\A_i\}_{i=1}^n}$

\paragraph{$\sem{\role A,\phi\triangleright\emptyset,\atom,\{F_i:
    \ite{s_i\doteq t_i}{\A_i^+}{\A_i^-} \}_{i=1}^n}$ where $n>0$}\label{rule:ite}
All the $\A_i$ start with a condition. We require that there are
recipes $r_1$ and $r_2$ such that $F_i(r_1)=t_i$ and $F_i(r_2)=s_i$
for each $i\in\{1,\ldots,n\}$. (If not, this is a specification error,
because it is unclear what $A$'s next step is.)  We define\\
$\mathcal{T}^+=\sem{\role A,\phi\triangleright\emptyset,\atom,\{F_i:\A_i^+\}_{i=1}^n}$
and\\
$\mathcal{T}^-=\sem{\role A,\phi\triangleright\emptyset,\atom,\{F_i:\A_i^-\}_{i=1}^n}$. \\
The translation
is:\\
$\ite {r_1\doteq r_2}{\mathcal{T}^+}{\mathcal{T}^-}$

\paragraph{$\sem{\role A,\phi\triangleright\emptyset,\atom,\{F_i:\C_i \}_{i=1}^n}$ where $n>0$}\label{rule:unlock}
Finally, if all the $\A_i$ conclude the atomic section, then the
translation is:\\
$\unlock.\sem{\role A,\phi\triangleright\emptyset,\chor,\{F_i:\C_i \}_{i=1}^n}$

\paragraph{$\sem{\_}$ for any other translation state}\label{rule:err} this is an
error, because it is unclear what $\role A$ should do next.

\subsubsection{Example}

Let us continue the example choreography from Fig.~\ref{fig:ex:ext}
and let us look at the translation for the role $\role
s$. The knowledge of $\role s$ gives us the frame:\\
$F_0=[\mathsf{\X_1\mapsto A,\X_2\mapsto B,\X_3\mapsto s,\X_4\mapsto
  \sk(s),}$\\
$\phantom{F_0=[}\mathsf{\X_5\mapsto \inv(\sk(s)),\X_6\mapsto
  \shk(A,s),\X_7\mapsto\ek(B)}]$ and we compute
$\sem{\role s,\emptyset\triangleright\emptyset,\chor,\{F_0:\C\}}$
where $\C$ is the entire choreography. The choreography begins with
atomic actions of $\role A$, checking if the key of $B$ is known,
asking $\role s$ if not and starting the main protocol otherwise. Thus
using rule~\ref{rule:atomic:other} we get a split into two
possibilities
$\{(F_0:\role A\rightarrow \role s:(\mathit{key},\role
B,N).\ldots),(F_0:\C_0)\}$ where $\C_0$ is the initial choreography
example of~Fig.~\ref{fig:ex:choreo} (with $\ek(\role B)$ replaced by
the variable $\mathit{EKB}$ that represents the key that $\role A$ has
looked up from memory). Now $\C_0$ starts with another atomic section of $\role A$
(the choice to either send a message or a key update). So we apply
again rule~\ref{rule:atomic:other} to split that possibility into two:
$
\begin{array}{l}
  (F_0:\role A\rightarrow \role s:(\mathit{key},\role B,N).\ldots),\\
  (F_0:\role A \rightarrow \role s: \scrypt((\role B,{\color{blue}\crypt(\ldots,\mathit{EKB})}),\shk(\role A,\role s)).\ldots)\\
  (F_0: \role A \rightarrow \role s: \scrypt((\role B,{\color{blue}\crypt(\ldots,\mathit{EKB})}),\shk(\role A,\role s)).\ldots) 
\end{array}$
Now finally all messages are something the server can receive, so with
rule~\ref{rule:receive} we get updated frames $F_1,F_2,F_3$ augmenting
$F_0$ with a new label $\X_8$, which is bound to the respective incoming message.

The checks that we can do for $F_1$ are that $\X_8$ is a triple, that
the first item is constant $\mathit{key}$ and the second item is
$\X_2$.\footnote{More realistically, $\role s$ should not expect a
  particular name $\role B$ but rather have it determined through
  $\role A$'s request. This is why we like to model $\ek(\cdot)$ as a
  public function (one can look up the public key of any role), but
  here we deliberately made this function private, so that $\role A$
  has to ask the server for the role. Anyway the semantics ensures
  that for any instantiation of the role variables with agent names,
  we have any number of server instances, so this comes without loss
  of attacks.} For ease of notation, assume we have functions
$\mathit{vtuple}_n$ and $\pi_n$ (for all $n\in\mathbb{N}$) with the
property $\mathit{vtuple}_n(t_1,\ldots,t_n)=_E\top$ and
$\pi_i(t_1,\ldots,t_n)=_E t_i$. Thus the checks for $F_1$ are
$\phi_1= \{\mathit{vtuple}_3(\X_8)\doteq \top,\pi_1(\X_8)\doteq
\mathit{key}, \pi_2(\X_8)\doteq X_2\}$.

In $F_2$ and $F_3$ we can check that symmetric decryption of $\X_8$
with key $\X_6$ succeeds and yields a pair. The blue parts in $F_2$
and $F_3$ cannot be decrypted, so $\role s$ cannot further check
anything about this. In our example theory $E$ below we have operators
$\scrypt$, $\vscrypt$ and $\dscrypt$ with the properties
$\dscrypt(\scrypt(m,k),k)=_E m$ and
$\vscrypt(\scrypt(m,k),k)=_E \top$.  Together they model AEAD
symmetric schemes, i.e., the attacker cannot modify the encrypted
message $m$ by modifying the ciphertext $\scrypt(m,k)$, as this would
lead to errors when decrypting; thus the decryption is an operation
that fails when applied to an incorrect message or the wrong key, and
we model that in the algebra by two functions, one telling us whether
decryption works with the given key and one that in the positive case
gives the result. We thus have in $F_2$ and $F_3$ the complete
set\footnote{\added{In Section~\ref{sec:mec} we give a complete set of
  checks procedure for an example theory $E$. This
  procedure can be used to verify that this $\phi_1\cup\phi_2$---plus
  some checks on the initial frame $F_0$ (like: $\X_4$ is the public
  key to $\X_5$) that we omitted here---is a complete set of checks.}}
of
the checks: $\phi_2=\phi_3=$ $\{\vcrypt(\X_8,\X_6)\doteq \top,$\\
$\mathit{vtuple}_2(\dscrypt(\X_8,\X_6))\doteq\top,
{\pi_1(\dscrypt(\X_8,\X_6))\doteq\X_2}\}$.

We thus have the translation\\
$\receive(\X_8).\sem{\role
  s,\emptyset\triangleright\phi_1\cup\phi_2,\chor,\{(F_1\text{:\ldots}),(F_2\text{:\ldots}),(F_3\text{:\ldots})\}}$

Applying rule~\ref{rule:check} several times to process all checks, we
get the possibilities partitioned (because $\phi_1$ only holds in
$F_1$ and $\phi_2$ only holds in $F_2$ and $F_3$) as follows:\\
$
\begin{array}{l}
  \mathsf{if}~\phi_1~\mathsf{then}~\sem{\role
  s,\phi_1\triangleright\emptyset,\chor,\{(F_1\text{:\ldots})\}}~\mathsf{else}\\
  \mathsf{if}~\phi_2~\mathsf{then}~\sem{\role
  s,\phi_1\triangleright\emptyset,\chor,\{(F_2\text{:\ldots}), (F_3\text{:\ldots})\}}~\mathsf{else}~\inact
\end{array}$
   
Let us just look at the most interesting branch, namely the positive
case under $\phi_2$. Here the next step in $F_2$ is
$ \role s \rightarrow \role B: \sign((\role
A,{\color{blue}\crypt((\msg,M),\mathit{EKB})}),\inv(\sk(\role s)))$
and in $F_3$ the corresponding step but with content $(\upd,K)$. So we
need to apply rule~\ref{rule:send} which requires a recipe $r$ that
works in both cases. Let
$r_{\color{blue}b}=\pi_2(\dscrypt(\X_8,\X_6))$ which gives the ``blue
message part'' that $\role s$ cannot decrypt in either $F_2$ or
$F_3$. Now $r=\sign((\X_1,r_{\color{blue}b}),\X_5)$ is the recipe that
works
in both cases. Thus the complete translation for the server role is:\vspace{5px}\\
$\begin{array}{l} \receive(\X_8).
   \mathsf{if}~\phi_1~\mathsf{then}~\send(\sign((\mathit{key},\X_2,\X_7)),\X_5)~\mathsf{else}\\
   \mathsf{if}~\phi_2~\mathsf{then}~\send(\sign((\X_1,\pi_2(\dscrypt(\X_8,\X_6))),\X_5))
\end{array}$


\subsection{Attack Semantics}
In this section, we formalize our notion of security. In a sentence, we consider there to be an attack if 
the system can possibly develop in a way that falsifies a given query over traces.

A \emph{security query} is built from the following grammar:
\begin{displaymath}\small
\begin{array}{r c l}
p &::=& \mathsf{event}(t)
        \mid \mathsf{event}(s) \sqsubseteq \mathsf{event}(t)
        \mid \mathsf{intruder}(t)
        \mid s \doteq t \\
  &\mid& \bot \mid p_1 \Longrightarrow p_2 \mid \forall X.\ p'
\end{array}
\end{displaymath}
The other logical operators can be added as syntactic sugar (in particular we will use $\land$).
We only consider a query well-formed if it contains no free variables. 

We characterize a trace, $(L_0,F_0,\mu)\stackrel{t_1,\ldots,t_n}{\to^*}(L_m,F_m,\mu)$, by the list of emitted events 
and the final knowledge of the intruder: $([t_1,\ldots,t_n],F_m)$.

Following is the semantics for evaluating a query:  
\begin{displaymath}
\begin{aligned}
([t_1,\ldots,t_n],F) &\models \mathsf{event}(s) \text{ iff } t_i \algeq s \added{\text{ for some }i}\\
([t_1,\ldots,t_n],F) &\models \mathsf{intruder}(s) \text{ iff } F(r) \algeq s \added{\text{ for some }r}\\
([t_1,\ldots,t_n],F) &\models s \doteq t \text{ iff } s \algeq t \\
([t_1,\ldots,t_n],F) &\models \bot \text{ never}\\
([t_1,\ldots,t_n],F) &\models  p_1 \Longrightarrow p_2  \\
   \text{iff} &\ ([t_1,\ldots,t_n],F) \not\models  p_1\\
  \text{or} &\ ([t_1,\ldots,t_n],F) \models p_2\\ 
([t_1,\ldots,t_n],F) &\models \forall X.\ p \\
\text{iff}&\text{ for all } s \text{ we have } ([t_1,\ldots,t_n],F) \models p[s/X]\\
([t_1,\ldots,t_n],F) &\models \mathsf{event}(s) \sqsubseteq \mathsf{event}(t) \\
  \text{iff}&
  \text{ for \replaced{each}{every} }t_i \algeq s \text{\added{,} there is a \replaced{distinct}{unique} }t_j \algeq t
\end{aligned}
\end{displaymath}
We consider a configuration as secure with regard to a given query if the query is valid on all 
traces from the configuration.


For the rest of this section, we will show how to encode the security goals in a choreography as queries. 

First, we modify the grammar of choreographies to allow the emission of events: 
  \begin{displaymath}\small
    \begin{array}{rcl@{\qquad}l}
      \C &::=& \ldots \\ 
      \A  &::=&  \ldots \\
          &\mid& \event(t).\A\\
    \end{array}
  \end{displaymath}
These events will be handled by the projection in the same way as sends: 
We check if we can find a recipe that produces $t$ in all frames, and 
return a translation error if not.

\added{
A secrecy goal expresses that a certain term should be kept a secret between a given set of agents.
That is, there is an attack if the intruder is not a member of the set and can produce the secret term.
The following transformation checks for this by having each member of the set emit a secrecy event 
at the end of each protocol run, including the names of the other agents and the term in the event.
}

To each secrecy goal $g = ``\secret{t}{\role A_1,\ldots,\role A_n}"$, we assign a unique \emph{event name} $e_g \in \Sigma \setminus \Sigma_p$.
We then replace each $g$ with
\begin{tabbing}
$\role A_1: \mathsf{event}(e_g(t,\role A_1,\ldots,\role A_n)).$\\
$\ \ \ \ \ \ \ \ \vdots$\\
$\role A_n: \mathsf{event}(e_g(t,\role A_1,\ldots,\role A_n))$
\end{tabbing}
The goal holds iff the following query holds: 
\begin{displaymath}
\forall X, \role A_1, \ldots,\role A_n.\ \neg \left(
    \begin{array}{l}
        \role A_1 \not\doteq \i \land \ldots \land \role A_n \not\doteq \i\ \land \\
        \mathsf{event}(e_g(X,\role A_1,\ldots,\role A_n))\ \land\\
        \mathsf{intruder}(X)
    \end{array}
    \right)
\end{displaymath}

\added{
Inspired by the hierarchy of authentication specifications of \cite{Lowe-Hierarchy}, 
we permit the user of CryptoChoreo to specify injective and noninjective authentication goals. 
Here, an attack on the noninjective authentication of agent $\role a$ to another agent $\role b$ on term $t$
would be if $\role a$ finishes the protocol believing that $\role b$ has played the same protocol,
but $\role b$ has either not played the protocol with $\role a$ or has done so with another term than $t$.
Injective authentication would also have an attack if $\role a$ can accept the value $t$ more often than $\role b$ commits to it.
We check authentication by inserting start and end events including the names of both parties and the value they authenticate on.
The authenticated party commits as early as possible to the value by emitting a start event and the 
authenticating party emits a corresponding end event at the very end of the protocol. 
For noninjective authentication, we then verify that the presence of an end event in a trace implies 
the presence of the corresponding start event. 
For injective authentication, we do the same, except that there must be a distinct start event for each end event.}

\added{
Our authentication goals give the guarantee to the authenticating party that the authenticated party has intended to use the 
term for the given protocol, but, unlike in \cite{Lowe-Hierarchy}, does not guarantee that they have reached the end of the protocol.
One could, alternatively, place the start event right before the last message to the authenticating party, 
though this would give spurious attacks if the protocol ends with messages that are not really a part of the authentication mechanism 
(for example if an agent ends by sending an end-signal in plaintext). 
We have found our version to be a pragmatic choice, that still usually has an attack if a meaningful attack exists in the stricter version. 
For some protocols, however, it might also be desirable to verify that the authenticated party has actually reached some given point in the protocol,
which can be checked in CryptoChoreo by placing the events manually. 
}

To a \added{noninjective authentication goal,} $g = ``\niauth{\role B}{\role A}{t}"$, 
we associate unique start- and end-event names $e_{gs},e_{ge} \in \Sigma \setminus \Sigma_p$.
We then replace the goal with the end event $\role B: \mathsf{event}(e_{ge}(t,\role A,\role B))$.
We want to check that the occurrence of this event implies the occurrence of a corresponding start event 
from $\role A$. However, it is not always possible to insert that start event right at the beginning of the choreography;
$t$ might contain values that have been generated during the protocol run. 

To solve this, we make the following modification to the projection semantics for $\role A$:\\
If we are computing $\sem{\role A,\phi\triangleright\emptyset,\chor,\{F_i:\C_i\}_{i=1}^n}$, 
$\role B:\mathsf{event}(e_{ge}(t,\role A,\role B))$ is in one of the $\C_i$, 
and the corresponding start event has not yet been generated, try the following:
If there is a recipe $r$ such that $F_1(r) = t \land \ldots \land F_n(r) = t$, 
return $\mathsf{event}(e_{gs}(r,\role A,\role B)).\sem{\role A,\phi\triangleright\emptyset,\chor,\{F_i:\C_i\}_{i=1}^n}$
while remembering that the event was generated.
Otherwise, continue computing $\sem{\role A,\phi\triangleright\emptyset,\chor,\{F_i:\C_i\}_{i=1}^n}$ as normal. \\
The goal is enforced by the following query: 
\begin{displaymath}
\forall X, \role A, \role B.\ 
    \begin{array}{l}
        \role A \not\doteq \i \land \role B \not\doteq \i \land \mathsf{event}(e_{ge}(X,\role A,\role B))   \\
        \ \ \ \ \Longrightarrow \mathsf{event}(e_{gs}(X,\role A,\role B)) 
    \end{array}
\end{displaymath}

For each \added{injective authentication goal,} $g = ``\auth{\role B}{\role A}{t}"$, we do the same procedure, except that we use the query:
\begin{displaymath}
\forall X, \role A, \role B.\ 
    \begin{array}{l}
        \role A \not\doteq \i \land \role B \not\doteq \i \Longrightarrow \\
        \mathsf{event}(e_{ge}(X,\role A,\role B))  \sqsubseteq \mathsf{event}(e_{gs}(X,\role A,\role B)) 
    \end{array}
\end{displaymath}

\added{
For more complicated properties, or if a custom placement of events is desired (for example if one wants an agent to emit a start event later in the protocol),
CryptoChoreo permits one to manually place events and specify custom queries. 
Algorithm~\ref{alg:asw} demonstrates both ways of specifying goals. 
}

\section{\replaced{Automation}{Mechanization} of the Projection}\label{sec:mec}
We give now an overview of \replaced{algorithms for the algebraic problems underlying the projection semantics,}{the mechanization} for a representative
algebraic theory.
\confver{
\replaced{ We here omit most details and all the proofs; these can be found in the extended version at \placeholder}{ most details are found in Appendix}.
}
\extver{
This section provides a simplified version, omitting many of details, which can be found in Appendix~\ref{app:mech}. 
We have included brief explanations of the correspondence between the results here and the results in the appendix. 
}

To implement the projection of Section~\ref{sec:projection}, one must
be able to do the following:
\begin{itemize}
  \item (\textit{word problem}) Given two terms $s$ and $t$, check if $s \algeq t$.
  \item (\textit{recipe composition}) Given a set of frame-term pairs $\{(F_1,t_1),\ldots,(F_n,t_n)\}$, decide if there is a recipe $r$ so  $F_1(r) \algeq t_1 \land \ldots \land F_n(r) \algeq t_n$, and return it if so.
  \item (\textit{complete set of checks}) Given a frame $F$ calculate a complete set of checks, i.e., a finite set of checks that implies all checks that can be made at all.
\end{itemize}

\added{
The projection semantics are defined to be agnostic to the precise algebraic theory, and in general these problems are not recursively computable.
There are, however, practical theories of cryptographic operators for which we can solve these problems.}
In this section\replaced{, we demonstrate this by giving procedures}{we will demonstrate how to do this} for a particular cryptographic model with Diffie-Hellman keys, symmetric encryption, and asymmetric encryption.
We assume that for all destructors we have a corresponding verifier that can be used to check whether \replaced{the destructor}{decryption} would succeed.

Recall that we are distinguishing terms, denoted $s,t\added{,\ldots}$,
\deleted{as used in the choreography}
\added{from recipes, denoted $r,r_1,r_a,\ldots$}.
On the
local behavior level, we have recipes that are built over $\Sigma_p$
and labels $\X_i$, where $\Sigma_p$ contains only public function
symbols, including destructors and verifiers.
\replaced{In the previous sections, terms occurred on the choreography level, and were}{They}
built over the alphabet $\Sigma$ and variables
$\mathcal{V}$, where $\Sigma$ \replaced{did}{does} not contain any destructors or
verifiers, but \replaced{could}{may} contain private functions. 
\added{
In this section, we additionally allow terms to contain 
destructors, as we here make precise how a term like $F(r)$ (which may contain destructors)
is equal in the $E$-theory to one that does not. 
For distinction, we call terms \emph{constructive} if they do not 
contain any destructors or verifiers. 
}

\subsection{Algebraic Theory}
As part of our algebraic theory we consider the following built-in
symbols:
\begin{displaymath}
  \begin{array}{l|lll}
    &\text{Constructors}
    &\text{Destructors}
    &\text{Verifiers}\\
    \hline
    \text{Tuples} & \pair & \fst,\snd & \vpair\\
    \text{Asymmetric Enc.} & \crypt & \dcrypt & \vcrypt\\
    \text{Private Keys} & \inv & \pubk & \vinv \\
    \text{Symmetric Enc.} & \scrypt & \dscrypt & \vscrypt\\
    \text{Signatures} & \sign & - & \vsign \\
    \added{\text{Signatures (cont.)}} & - & \added{\mathit{open}}  & \added{\mathit{vopen}} \\
    \text{Diffie-Hellman} & \cexp& \expinv & \vexp \\
  \end{array}
\end{displaymath}
All the constructors are part of $\Sigma$ and can occur in terms, and
all except $\inv$ are public and thus in $\Sigma_p$ and can occur in
recipes. Note that we had earlier used $n$-tuples for simplicity, and
here have only binary tuples, but this can be seen as syntactic
sugar. Besides these symbols, the modeler can declare other further
function symbols that can be either public (and thus both part of
$\Sigma$ and $\Sigma_p$), e.g., to model hash functions or public
constants, or that can be private (and thus only part of $\Sigma$),
e.g., to model key infrastructures or fixed secrets between
agents. However, these user-defined functions cannot have any algebraic
properties. We also have a public constant $\top$ both in $\Sigma$ and
$\Sigma_p$. We call a recipe \emph{constructive} if it does not
contain destructors or verifiers.

\begin{definition}\label{definition:e:algebra}
We define our algebra $E = R \cup B$ where

\noindent
$R =$

$\{\vpair(\pair(x,y),\top) \doteq \top$,

$\fst(\pair(x,y),\top) \doteq x$,  $\snd(\pair(x,y),\top) \doteq y$,

$\vscrypt(\scrypt(x,y),y) \doteq \top$,

$\dscrypt(\scrypt(x,y),y) \doteq x$,

$\vcrypt(\crypt(x,y),\inv(y)) \doteq \top$,

$\dcrypt(\crypt(x,y),\inv(y)) \doteq x$,

$\vsign(\sign(x,\inv(y)),y) \doteq \top$,

$\added{\mathit{vopen}(\sign(x,y),\top) \doteq \top}$,

$\open(\sign(x,\deleted{\inv(}y\deleted{)}),\replaced{\top}{y}) \doteq x$,

$\vinv(\inv(x),\top) \doteq \top$,
$\pubk(\inv(x),\top) \doteq x$,

$\vexp(\cexp(x,y),y) \doteq \top$, $\expinv(\cexp(x,y),y) \doteq x\}$,

\noindent
and
$B = \{\cexp(\cexp(x,y),z) \doteq \cexp(\cexp(x,z),y)\}$.

Let $\algeq$ denote the congruence induced by these equations and
$=_B$ the congruence induced just by the equation in $B$.
\end{definition}

The functions $\vpair$, $\fst$, $\snd$, $\vinv$, and $\pubk$ should
actually be unary functions, because they do not require a key. For
uniformity, we have made them binary functions like all the other
destructors and verifiers, and we use $\top$ as a dummy value for the
key-position.
\added{
As all destructors can now be regarded as decryption operators
we make no distinction between the concepts in the context of this
algebraic theory.
}

The reader may be surprised to see a verifier for Diffie-Hellman
exponentiation. This is because our method below requires that every
destructor has a corresponding verifier.
However, the verifier $\vexp$
exists only pro forma: if \deleted{we}our procedure runs into a situation where
it actually employs $\vexp$, it stops with an error.
The only
situation where it would be employed is\deleted{,} if we have an agent \added{who} knows
both $x$ and a term (equivalent to) $\cexp(t,x)$, but not $t$, and
this does not occur in standard uses of Diffie-Hellman.
\added{
Thus, $\vexp$ is a tool we use for our proofs, while still preserving soundness: 
by aborting translation we do not give this unrealistic capability to honest agents, 
and thus a successful translation is correct if we drop
the unrealistic equation $\vexp$.}

A similar question may arise from the destructor and verifier for
private keys. Many approaches model a public constructor that from a
given private key generates a public key; we use here instead a
private constructor $\inv$ to map a public key to a corresponding
private key; this allows us easily model public-key infrastructures
like $\ek(\role A)$ being the public encryption of $\role A$ where
$\ek$ is a public function (so every agent can lookup keys). The small
price to pay is that the inverse mapping from private to public key is
called a destructor and that we have a verifier to check if a private
key really fits with the public key.

\subsection{\added{Algorithms}}

\added{We now sketch out how the three problems mentioned at the
  beginning of this section---the word problem, the recipe composition
  problem, and the complete set of checks problem---are all
  computable. The algorithms for this are close to standard protocol
  analysis methods, and not a main contribution of this
  work. Moreover, basic requirements of efficiency (e.g., avoiding
  repeated analysis steps and checks) make the algorithms rather
  involved with details.}
\confver{\added{We thus give the precise algorithms and proof
  of correctness only in the extended version of this
  paper, available at \placeholder.}}
\extver{The full details and proofs are available in Appendix~\ref{app:mech}.}

\subsubsection{\added{Word Problem}}

The considered theory $E=R\cup B$ allows us to decide the word
problem, i.e., for terms or recipes $s,t$, whether $s=_E t$. This is
because $R$ used as rewrite rules modulo $B$ (i.e.,
$\rightarrow_{R/B}$) is convergent. Thus we only need to compare the
normal forms of $s$ and $t$ modulo $B$. The equivalence class modulo
$B$ of any term is finite and easily computable. \deleted{See
  Theorem~1
  in the appendix.} 
\extver{
  In Appendix~\ref{app:mech}, this is proven in Theorem~\ref{thm:word:problem}
}

\subsubsection{\replaced{Constructive Recipe Composition}{Compose}}
\added{We now first solve a simplified version of the recipe
  composition problem, where recipes do not contain destructors or verifiers.}
We define a function $\mathit{compose}(F,t)$ that, given a frame $F$ and a
term $t$, obtains all \emph{constructive} recipes $r$ such that $F(r)=_B
t$. Roughly, for every term $t_0$ in $[t]_B$ (the equivalence class of
$t$ modulo $B$) we can check if $t_0$ is a term in the frame, and
additionally, if $t_0=f(t_1,\ldots,t_n)$ for a public $f$ (which
cannot be a destructor or verifier by construction), we
recursively compute $\mathit{compose}(F,t_i)$ and from the results construct
the solutions for $\mathit{compose}(F,t)$ as expected.

\added{For example, let $F=[\X_1\mapsto X, \X_2\mapsto \cexp(g,Y)]$
  (where $g$ is a public constant) and the target is
  $t=\cexp(\cexp(g,X),Y)$. The equivalence class is $[t]_B=\{t,t'\}$
  with $t'=\cexp(\cexp(g,Y),X)$. Neither $t$ nor $t'$ is directly
  contained in $F$, so we recursively check the subterms: for $t$ we
  compute $\mathit{compose}(F,\cexp(g,X))=\{\cexp(g,\X_1)\}$ and
  $\mathit{compose}(F,Y)=\emptyset$; the second call fails because we cannot
  obtain $Y$. For $t'$ we compute
  $\mathit{compose}(F,\cexp(g,Y))=\{\X_2\}$ and
  $\mathit{compose}(F,X)=\{\X_1\}$, and we thus get
  $\mathit{compose}(F,t)=\{\cexp(\X_2,\X_1)\}$.}

\extver{In Appendix~\ref{app:mech}, we show that $\mathit{compose}$ returns all constructive recipes for a given term in Lemma~\ref{lemma:compose-complete}.}

\subsubsection{Analysis}

We further introduce the notion of an \emph{analyzed frame}, i.e.,
where the frame contains every term that can be obtained using a
destructor on any message in the frame, using a constructive recipe
for the key term.

We define an analysis procedure that successively applies decryption
steps as long as possible: for every message that potentially can be
decrypted, we check if we can compose the decryption key. If so, the
resulting message is added to the frame. Whenever we add an analyzed
message to the frame, we also need to check again all those messages
for which we previously did not have the decryption key. \confver{\replaced{In
  the extended version we show that this terminates and produces
  correct analyzed frames.}{We show that
  this terminates (Theorem~2
) and produces correct
analyzed frames (Theorem~3
).}}

\added{For example, the frame
  $F=[\X_1\mapsto X, \X_2\mapsto \cexp(g,Y),$}

\noindent
\added{$\X_3\mapsto
  \scrypt((A,\crypt(M,\pk(B))),\cexp(\cexp(g,X),Y))]$ is not analyzed
  because there is a constructive recipe for the Diffie-Hellman key of
  the message in $\X_3$, as seen before. We thus add to the frame
  $[\X_4\mapsto (A,\crypt(M,pk(B)))]$ and note that $\X_4$ is a
  shorthand for $\dscrypt(\X_3,\cexp(\X_2,\X_1))$. This does not change
  the derivable messages of the frame of course. The resulting frame
  is still not analyzed as $\X_4$ is a pair, and we can add
  $[\X_5\mapsto A,\X_6\mapsto\crypt(M,pk(B))]$ and note that $\X_5$
  and $\X_6$ are shorthands for $\pi_1(\X_4,\top)$ and $\pi_2(\X_4,\top)$, respectively. This is now analyzed
  since we have no constructive recipe for $\pk(B)$ (and in fact no
  recipe at all).} 

\extver{In Appendix~\ref{app:mech}, we show that this procedure terminates in Theorem~\ref{thm:ana:terminate}
and that produces correctly analyzed frames in Theorem~\ref{thm:ana:correct}.}

\replaced{The compose and analyze algorithms together solve the recipe
  composition problem: Suppose $F(r)=_E t$ where $t$ does not contain
  destructors and verifiers, and $F'$ is the analyzed version of $F$,
  then $\mathit{compose}(F',t)\neq\emptyset$. The proof is essentially that we
  look at any application of a destructor or verifier in $r$ (if there
  is any) that has no further destructor or verifiers as subterms, say
  $d(r_0,r_1)$ where $r_0,r_1$ are constructive. Then the analysis
  must have found it, i.e., we have a label $\X$ in $F'$ that produces
  a term that is $E$-equivalent to $F'(d(r_0,r_1))$. In this way we
  can successively replace all destructors in $r$ by shorthands of
  $F'$ until we obtain a recipe that is constructive and thus found by
  $\mathit{compose}$. 
}{In an analyzed frame $F$, $\mathit{compose}$ finds a
recipe for every term $t$ that can be obtained with any recipe
(Lemmas~4
~and~5
).}

\extver{In Appendix~\ref{app:mech}, Lemma~\ref{lemma:compose-in-analyzed} shows that if there is any recipe for producing a
term in an analyzed frame, then $\mathit{compose}$ will also return such a recipe.}

\subsubsection{\added{Complete Set of }Checks}

\deleted{Finally, given an analyzed frame, we show how to compute a complete
set of checks. In particular, it is sufficient to restrict oneself to
constructive recipes for checks in an analyzed
frame~(Lemma~6
and
Theorem~4
).}

\added{Our analysis procedure also computes some checks:
  whenever we have $\X\mapsto t$ in a frame $F$ such that $t$ can be
  decomposed, say $d(t,k)$ is a redex and $k$ can be constructed in
  the $F$, then we can also check $v(t,k)\doteq\top$ for the
  corresponding verifier.
  }

\added{For example, in the above analysis example we would derive the
  checks $\vscrypt(\X_3,\cexp(\X_2,\X_1))\doteq\top$ and
  $\vpair(\X_4,\top)\doteq\top$.}

\added{In general, however, this does not yet give the complete set of
  checks, for instance $F=[\X_1\mapsto h(N), \X_2\mapsto N]$ is
  already analyzed, but there is still the check $\X_1\doteq
  h(\X_2)$. A complete set of checks can now be found by checking for
  every label (like $\X_1$ here) if there is a different way to
  construct it. Essentially, the proof is that, given an analyzed
  frame and an arbitrary check $r_1\doteq r_2$, we can reduce it to
  constructive recipes, and where one of the sides is a label.}

\extver{In Appendix~\ref{app:mech}, Theorem~\ref{thm:complete-checks} shows how a finite complete set of checks can be derived for any frame.}

\subsubsection{\added{Extension to Multiple Frames}}

Finally, the compose procedure can be extended to the recipe
composition problem, i.e., given analyzed and checked $F_1,\ldots,F_n$
and goal terms $t_1,\ldots,t_n$, find a single recipe $r$ with
$F_i(r)=t_i$ for all $i$, because we try to get  a solution for
$F_1(r)=t_1$ and if it exists, then it works in all frames, because
they are checked\deleted{ (Theorem~5
)}.

\extver{In Appendix~\ref{app:mech}, Lemma~\ref{lemma:check-in-alternative} shows that if a frame is fully checked it does not matter 
which recipe from $\mathit{compose}$ one picks (in any actual run, they will all produce the same term). 
Theorem~\ref{thm:compose:correct} shows how to obtain a single recipe that works for all given frame-term pairs.}




\section{Case studies}\label{sec:case:studies:meta}

\subsection{Exporting to ProVerif}\label{sec:proverif}
In this section, we summarize how the output from the projection
semantics of Section~\ref{sec:projection} can be further translated to
ProVerif code for automatic verification.  More details can be found
in \confver{\replaced{the extended version at \placeholder}{Appendix Section }.}\extver{Appendix~\ref{app-sec:proverif}.}
We start by unfolding the labels from the initial knowledge in each
local behavior, similarly to what we do in the semantics of
Figure~\ref{fig:semantics:local}.
Each local behavior is almost a valid ProVerif process already, except
for the use of nondeterministic choice and memory cells.

Nondeterministic choice is simple to encode: we take a message from
the network (the intruder) indicating which branch to take, and then
branch on the content of that message.  Thus, we leave it to the
intruder to pick the branch that will lead to an attack (if one
exists).

Encoding memory cells is more involved\added{, since reasoning about
long-term mutable state in ProVerif is a well-known
difficulty for which several extensions and front-ends have been
proposed~\cite{Modersheim10AIF,BruniModersheimNielson15,ArapinisRitterRyan11,DBLP:conf/csfw/ChevalCT18,HessModersheimBruckerSchlichtkrull21}}.
We encode each memory cell as
a private channel, and ensure by construction that the channel always
contains exactly one message (except when the message has been
consumed in an atomic section that has not yet been left).  Then, in
the semantics of ProVerif, a read from the memory cell corresponds
exactly to reading the term that was most recently placed in the
channel.

However, when the ProVerif process is translated to Horn-clauses,
certain overapproximations will mean that it is no longer guaranteed
that a message on a private channel is consumed in the right order or
only once.  We decrease the chance of a false positive by adding a
counter to each memory cell, which is incremented on every write.
Adding the admissible axiom that if two values written to a memory
cell are associated with the same counter value they must be equal,
excludes many impossible models during the proof search.


\subsection{Examples}\label{sec:case:studies}
We implemented a tool \deleted{in Haskell} that automates the projection from choreographies to local behaviors,
based on the definitions in Sections~\ref{sec:projection} and~\ref{sec:mec}. 
\added{It is implemented in around 5000 lines of Haskell code.}
\replaced{The tool}{It} also supports the generation of a ProVerif file from these local behaviors, following the steps described in Section~\ref{sec:proverif}.

We will in the following describe one particular example, but have made more available. \deleted{All our examples combined verify in a minute when exported to ProVerif.} 
We wish to highlight the following notable examples:
\begin{itemize}
  \item \texttt{blind-forward.choreo} is similar to the example from Section~\ref{sec:choreo}. Here a trusted third party either helps $\role A$ authenticate a new encryption key with 
  or authenticate a request to $\role B$, without knowing which branch it is in.
  \added{We have verified that in the branch where $\role B$ gets a new public key from $\role A$, that key is authenticated, and that in the branch where data is
  sent to $\role A$ with the key we have secrecy of the data.}
  \item \texttt{SSO.choreo} describes a protocol where an agent $\role A$ authenticates to another agent $\role B$ using a trusted third party $\role{ttp}$\added{ to establish a secure channel}. 
  We can also verify this when the behavior of $\role{ttp}$ is taken from \texttt{SSO-API.choreo}, where $\role{ttp}$ is implemented like an API responding to queries 
  and saving their state using memory cells. 
  Our tool includes an option for selecting the behavior of participants from different choreographies like this.
  \added{
   We have verified that data sent over the channel is authenticated and secret, and that the symmetric key shared to establish the channel is secret.
  }
  \item We demonstrate in \texttt{tpm-simple.choreo} and \texttt{tpm-simple-API.choreo} the security of a TPM that can either declassify a given value 
  or delete it. In particular, this protocol is non-monotonic, as it should be transparent to the owner of the value which choice was taken, and that if one choice is made then the other cannot be made later.
  \added{
    We have verified that if anyone (the intruder) comes to know the value then the ``opened'' event was triggered by the TPM, and 
    that if the ``refused'' event was triggered by the TPM then the value cannot be known by the intruder.
  }
  \item \texttt{SSO-DH.choreo} and \texttt{SSO-DH-API.choreo} demonstrates a Diffie-Hellman exchange mediated by a trusted third party.
  \added{
    We have verified the secrecy of a value encrypted with the established Diffie-Hellman key. 
  }
  \item \added{\texttt{NSLPK.choreo} contains the Needham-Schroeder public-key protocol with Lowe's fix. We have verified injective authentication and secrecy in both directions.}
  \item \added{\texttt{tls-1.3.choreo} is a simplified model of TLS 1.3. We have verified injective authenticity and secrecy of the data transmitted over the established channel in both directions. }
  \item \texttt{ASW.choreo} contains the example described in the \replaced{following}{sections}. 
\end{itemize}

\begin{table}
\caption{}\label{table:example:stats}
\begin{tabular}{r | r | r | r}
  \added{\textbf{File}} & \added{\textbf{LOC}} & \added{\textbf{Translation}} & \added{\textbf{Verification}} \\ 
  \hline
  \added{\texttt{blind-forward.choreo}}     & \added{28}  & \added{0.06s} & \added{0.02s} \\
  \added{\texttt{SSO.choreo}}               & \added{32}  & \added{0.10s} & \added{1.70s} \\
  \added{\texttt{tpm-simple.choreo}}        & \added{34}  & \added{0.06s} & \added{0.02s} \\
  \added{\texttt{SSO-DH.choreo}}            & \added{28}  & \added{0.06s} & \added{117.20s} \\
  \added{\texttt{NSLPK.choreo}}             & \added{23}  & \added{0.05s} & \added{0.07s} \\
  \added{\texttt{tls-1.3.choreo}}           & \added{37}  & \added{0.12s} & \added{171.80s} \\
  \added{\texttt{ASW.choreo}}               & \added{101} & \added{0.19s} & \added{0.13s} \\
\end{tabular}
\end{table}

\noindent
\added{Table~\ref{table:example:stats} shows some additional information on the examples described above. The given runtimes were obtained on an HP EliteBook 840 G10.}
\added{
\textbf{LOC} (Lines Of Code) is the number of lines in the choreography file, excluding comments and empty lines. 
\textbf{Translation} is the time in seconds CryptoChoreo used to translate the choreography file to local behaviors and then to ProVerif code. 
\textbf{Verification} is the time in seconds ProVerif used to verify the obtained ProVerif code.}

\added{
In our experience, protocols that succeed verification usually do so in a few seconds or less, 
except when the protocol involves non-trivial uses of Diffie-Hellman (like \texttt{SSO-DH.choreo} and \texttt{tls-1.3.choreo}), 
in which case verification could take up to a few minutes. 
For protocols that fail verification the result are much more varied,
ranging from an immediate negative result, to a negative result after a couple of minutes, 
to the verification procedure seemingly going into nontermination. 
}

\subsection{\added{The Asokan-Shoup-Waidner Protocol}}\label{sec:case:studies:asw}

The Asokan-Shoup-Waidner (ASW) protocol~\added{\cite{DBLP:conf/sp/AsokanSW98}} serves as a motivating example that demonstrates the expressiveness of our choreography language, particularly its support for explicit branching, nondeterministic choice, conditional behavior, and long-term memory access.

ASW is a fair contract signing protocol involving three participants: an originator $\role O$, a responder $\role R$, and a trusted third party $\role{TTP}$. The protocol mainly ensures that either both parties obtain a binding contract, or neither does. 
The key challenges addressed by this protocol are:
\begin{itemize}
  \item \textbf{Timeouts and abort scenarios}: Participants may timeout, leading to different protocol continuations.
  \item \textbf{Stateful TTP}: The TTP must maintain memory of previous contract states to prevent inconsistent responses.
  \item \textbf{Conditional logic}: Protocol actions depend on checking stored memory values.
\end{itemize}

\added{
In the following, we use formats like $f_1$, $f_2$, $f_\mathit{abort}$ and $f_\mathit{resolve}$. 
These are transparent functions that structure messages but provide no cryptographic guarantees. 
A formatted message $f(t_1,\ldots,t_n)$ can be modeled as a tuple $(c_f,t_1,\ldots,t_n)$ where $c_f$ is a public constant.}

\begin{algorithm}[h!]
\caption{ASW Choreography}
\label{alg:asw}
\begin{algorithmic}[1]
\State $\role O : \nu \Text.\, \nu \NO.\;$
\State $\role O \rightarrow \role R:  \underbrace{\sign(f_1(\role
  O,\role R,\ttp,\Text,h(\NO)),\inv(\pk(\role O)))}_{=:M_1} .$ 

\State $\role R : {\color{red}\role R \rightarrow \role O :
  \timeout.\textsc{Abort}(\role O,M_1)}$
\State $ {\color{red}+} \;\;\nu \NR.\;$
\State $\quad\; \role R \rightarrow \role O : 
\underbrace{\sign(f_2(M_1,h(\NR)),\inv(\pk(\role R)))}_{=:M_2} .$ 
\State $\quad\; \role O : \color{red}\role O \rightarrow \role R :
\timeout.\textsc{Resolve}(\role R,M_1,M_2)$
\State $\quad\; {\color{red}+} \;\; \; \role O \rightarrow \role R : \NO .$ 
\State $\added{\qquad\quad \role R : \mathsf{event}(\mathit{obtain}(\role R,M_1)). \mathsf{event}(\mathit{finish}(\role R,M_1)).}$
\State $\qquad\quad \role R : \color{red}\role R \rightarrow \role O:
\timeout.\textsc{Resolve}(\role O,M_1,M_2)$
\State $\qquad\quad  {\color{red}+}\;\;\; \role R \rightarrow \role O : \NR $ 
\State $\added{\qquad\qquad\ \ \role O : \mathsf{event}(\mathit{obtain}(\role O,M_1)). \mathsf{event}(\mathit{finish}(\role O,M_1)).}$
\State $\qquad\qquad\ \ \added{\auth{\role O}{\role R}{M_2}}$
\State $\qquad\qquad\ \ \added{\niauth{\role R}{\role O}{M_1}}$
\State 
\State $\textsc{Abort}(\role A,M_1)=$
\State $\role A \rightarrow \role {\ttp} :
sign(\fabort(M_1),\inv(\pk(\role A)))$
\State $\role {\ttp} : \blank := \ttpmem[M_1].$ 
\State $\qquad \ttpmem[M_1] := \aborted .$ 
\State $\qquad \role {\ttp} \rightarrow \role A :
\sign(\fabort(M_1),\inv(\pk(\ttp)))$
\State $\added{\qquad \role A : \mathsf{event}(\mathit{finish}(\role A,M_1))}$
\State $+ \quad\,\aborted := \ttpmem[M_1].$
\State $\qquad \role {\ttp} \rightarrow \role A : \sign(\fabort(M_1),\inv(\pk(\ttp)))$
\State $\added{\qquad \role A : \mathsf{event}(\mathit{finish}(\role A,M_1))}$
\State $+ \quad\, (\resolved,M_2) := \ttpmem[M_1].$
\State $\qquad \role {\ttp} \rightarrow \role A :
\sign(\fresolve(M_1,M_2),\inv(\pk(\ttp)))$
\State $\added{\qquad \role A : \mathsf{event}(\mathit{obtain}(\role A,M_1)). \mathsf{event}(\mathit{finish}(\role A,M_1))}$
\State 
\State $\textsc{Resolve}(\role A,M_1,M_2)=$
\State $\role A \rightarrow \role {\ttp} :
sign(\fresolve(M_1,M_2),\inv(\pk(\role A)))$
\State $\role {\ttp} : \blank := \ttpmem[M_1].$ 
\State $\qquad \ttpmem[M_1] := (\resolved,M_2) .$ 
\State $\qquad \role {\ttp} \rightarrow \role A :
\sign(\fresolve(M_1,M_2),\inv(\pk(\ttp)))$
\State $\added{\qquad \role A : \mathsf{event}(\mathit{obtain}(\role A,M_1)). \mathsf{event}(\mathit{finish}(\role A,M_1))}$
\State $+\quad\, \aborted := \ttpmem[M_1].$
\State $\qquad \role {\ttp} \rightarrow \role A : \sign(\fabort(M_1),\inv(\pk(\ttp)))$
\State $\added{\qquad \role A : \mathsf{event}(\mathit{finish}(\role A,M_1))}$
\State $+\quad\, (\resolved,M_2) := \ttpmem[M_1].$
\State $\qquad \role {\ttp} \rightarrow \role A :
\sign(\fresolve(M_1,M_2),\inv(\pk(\ttp)))$
\State $\added{\qquad \role A : \mathsf{event}(\mathit{obtain}(\role A,M_1)). \mathsf{event}(\mathit{finish}(\role A,M_1))}$
\Statex
\Statex \added{\textbf{Verified Queries:}}
\Statex $\added{\forall M_1,M_2.}
\left(
  \begin{array}{l}
    \added{\mathsf{event}(\mathit{abort}(M_1)) \land} \\ \added{\mathsf{event}(\mathit{resolve}(M_1,M_2))}
  \end{array} 
\right) \added{\Longrightarrow \bot}$
\Statex
\Statex  
\Statex $\added{\forall M_1,A,B.}
\left(
  \begin{array}{l}
    \added{\role A \neq \role{ttp} \land \role B \neq \role{ttp} \land \role B \neq \role i  \land} \\ 
    \added{\mathsf{event}(\mathit{finish}(\role B, M_1)) \land} \\ 
    \added{\mathsf{event}(\mathit{obtain}(\role A, M_1))}
  \end{array}
\right) \added{\Longrightarrow}$
\Statex \added{$ \hfill \mathsf{event}(\mathit{obtain}(\role B, M_1))$}
\end{algorithmic}
\end{algorithm}

\subsection{Nondeterministic Choice and Branching}
Algorithm~\ref{alg:asw} presents the ASW choreography using explicit branching in the message flow. After $\role O$ sends $M_1$ to $\role R$, the responder can either timeout and trigger an abort request (shown in red), or continue by generating $\NR$ and replying with $M_2$.
In the continuation, timeouts may occur again: after receiving $M_2$, $\role O$ can either timeout and ask the $\role{\ttp}$ to resolve, or proceed by sending $\NO$; symmetrically, after receiving $\NO$, $\role R$ can either timeout and resolve, or complete by sending $\NR$. The explicit $+$-marked alternatives capture this nondeterminism directly at the level of the choreography.

\subsection{Conditional Behavior and Memory Access}
The distinctive feature of ASW is the behavior of the trusted third party, which consults and updates long-term memory cells indexed by the contract identifier $M_1$. In both \textsc{Abort} and \textsc{Resolve}, the $\role{\ttp}$ first reads $\ttpmem[M_1]$ (with $\blank$ denoting an uninitialized cell) and then branches on the stored value.

This illustrates several language features:
\begin{enumerate}
  \item \textbf{Memory read}: $\blank := \ttpmem[M_1]$ reads the current state associated with $M_1$\added{ and checks that it is currently uninitialized}.
  \item \textbf{Case distinction on memory}: depending on whether $\ttpmem[M_1]$ is still $\blank$, has been set to $\aborted$, or contains $(\resolved,M_2)$, the $\role{\ttp}$ returns a consistently matching signed response.
  \item \textbf{Memory write}: $\ttpmem[M_1] := \aborted$ (in \textsc{Abort}) and $\ttpmem[M_1] := (\resolved,M_2)$ (in \textsc{Resolve}) record the outcome so that repeated requests cannot lead to contradictory replies.
  \item \textbf{Non-monotonic memory}: the same cell can be read multiple times and updated across different interactions in the run.
\end{enumerate}

\subsection{Semantic Challenge: Participant Knowledge}
A subtle aspect of ASW mechanization arises during projection. When $\role O$ or $\role R$ make decisions (e.g., ``abort''), they must send signals that the $\role{ttp}$ later interprets. However, the originating participant does not see the entire memory state of the TTP. During projection, the analyzer must determine:
\begin{itemize}
  \item When to \emph{accept} incoming values without immediate verification (e.g., when $\role O$ receives an untagged nonce response).
  \item When to insert runtime \emph{checks} that compare received values against later constraints (e.g., when verifying $h(NO) = X$ after learning $NO$).
\end{itemize}

\subsection{\added{Verification}}
\added{
We have verified several queries on the protocol. 
Firstly, in the branch where $\role{ttp}$ is not needed for mediation, we can show mutual authentication on the contracts.
$\role R$ can only noninjectively authenticate $\role O$ on $M_1$ since $\role O$ does not sign $\NR$. 
We do not consider this crucial, in any case, since freshness would in practice be guaranteed by the text in the contract. 
The main properties of the protocol are checked by the two manual queries. 
Firstly, we want it to be impossible for $\role{ttp}$ to both resolve a contract positively and abort it. 
We therefore insert events to record when either happens, and then verify that at most one of them can have fired in any given trace.
Secondly, we want the protocol to be fair. We formalize this by adding the event $\mathit{obtain}(\role A, M_1)$ whenever an agent $\role A$ has a valid contract 
for the text in $M_1$ and $\mathit{finish}(\role A,M_1)$ whenever an agent considers the protocol to be finished (either because it was aborted or because they got the contract).
As long as the intruder does not block communication with $\role{ttp}$ forever, it is easy to see that each agent will be able to reach the $\mathit{finish}$ event. 
We then verify that if one agent has finished and the other has a valid contract, then the first must have a valid contract too. 
}

\section{Related Work}\label{sec:related:work}
\noindent\textbf{Alice-and-Bob notation.}
Several lines of work have given Alice-and-Bob notation a precise semantics by compiling them to lower-level role-based specifications, and by rejecting notations that are not executable because e.g. senders cannot construct a message or receivers cannot check it.
Early semantics were developed for the free term algebra and later extended to equational theories, often by reducing executability questions to intruder deduction and unification problems \cite{JRVP00,CaleiroBasinVigano06,Modersheim09}.
Our semantics follows this tradition, but lifts it from linear notations to \emph{choreographies} with nondeterministic choice, branching, and mutable long-term memory.
This combination is essential for modeling modern protocol interactions with stateful services and APIs, while maintaining the compact global view that motivates Alice-Bob notations.

\medskip
\noindent\textbf{Protocol models and tool support.}
At the other end of the spectrum, multi-set rewriting languages (e.g., as used by Tamarin and the AVISPA family) provide an explicit account of state and message flows and are well-suited for reasoning with rich adversary models \cite{MeierSchmidtCremersBasin13,ArmandoEtAl05,BasinModersheimVigano05}.
Process-calculus based tools such as ProVerif offer a more program-like view of each role and highly automated verification, typically by a sound over-approximation \cite{Blanchet01,Blanchet16}.
Our contribution is complementary: CryptoChoreo aims to be a high-level \emph{specification} language that can be translated into such backends.
In this respect, our approach is aligned with lines of work that provide source-to-source translations or front-ends for existing verification tools, such as SAPIC/SAPIC+ \cite{KremerKunnemann16,ChevalEtAl22}, but with a focus on preserving the readability and single global story of a choreography.

\medskip \noindent\textbf{Choreographies and global types.}
Choreographic programming and the theory of multiparty session
types/global types study how a global description of multiparty
interaction can be projected to local behaviors, with correctness
guarantees such as deadlock freedom
\cite{CHY12, CM13, HYC16}.  Our work is inspired
by the same global-to-local methodology, but targets the Dolev-Yao
setting with an active adversary, cryptographic
constructors/destructors, equational theories, and explicit attacker
knowledge.

\medskip
\noindent\textbf{State and APIs in protocol models.}
Stateful extensions and encodings have been studied both at the specification level and in tool-oriented front-ends \cite{KremerKunnemann16}.
\added{Reasoning about long-lived global state is in particular a known difficulty for ProVerif's Horn-clause abstraction, which has motivated several targeted extensions and front-ends, including set-membership abstractions and AIF \cite{Modersheim10AIF}, Set-Pi \cite{BruniModersheimNielson15}, StatVerif \cite{ArapinisRitterRyan11}, GSVerif \cite{DBLP:conf/csfw/ChevalCT18}, and the PSPSP framework \cite{HessModersheimBruckerSchlichtkrull21}.}
CryptoChoreo makes state explicit at the choreography level via memory cells\added{, and our ProVerif export takes inspiration from these works to mitigate the same overapproximation issues at the translation level}.

\medskip
\noindent\textbf{Equational reasoning and \replaced{automation}{mechanization}.}
Reasoning about message construction, parsing, and checks in the presence of equational theories is a classic challenge in symbolic protocol analysis \cite{ChevalierVigneron02}.
\replaced{We give algorithms for}{Our mechanization focuses on} a representative theory combining standard constructors/destructors with Diffie-Hellman exponentiation, and yields an effective projection procedure \cite{SchmidtMeierCremersBasin12}.
This is intentionally backend-agnostic at the level of the core semantics, while our ProVerif export demonstrates one concrete and practical target.

\section{Conclusions}\label{sec:conclusions}
We introduced CryptoChoreo, a choreography language for cryptographic protocols that extends Alice-and-Bob notation with nondeterministic choice, conditional branching, and mutable long-term memory, together with a projection-based semantics\deleted{ and an effective mechanization for a representative equational theory.  Our ProVerif export shows that these high-level specifications can be connected to automated verification backends in practice.}\added{. The semantics is defined for an arbitrary algebraic theory $E$ and computes how honest agents execute the protocol, namely how they compose outgoing messages and decompose and check incoming messages.
  This offers the modelers a great tool to write specifications, because they just specify how the messages look like in an unattacked protocol run, and the translation semantics figures out which steps the actual implementation has to do.
  We can thereby avoid many specification mistakes (where a part of a message that should be checked is forgotten in the implementation). 
  Thus, this is a basis for generating secure-by-design implementations where, simply, functions such as encryption and decryption need to be connected to real cryptographic functions etc., 
  and where the implementation is in a one-to-one relationship with the formal model fed into verification tools. In fact, we plan to investigate if it can also be used with cryptographic verification tools like CryptoVerif.}

\added{While for a general theory $E$, the translation is not computable, we give the procedure for an example theory and the backend ProVerif. The example theory covers the usual constructor/destructor theories as well as Diffie-Hellman; here we use a verifier that could not be implemented in reality, but that allows to integrate exponentiation uniformly with the rest of the method based on analyzing all terms as far as possible.}

\added{We initially considered translation to SAPIC{+}~\cite{ChevalEtAl22} because it is similarly on the level of process calculus and has translators to both ProVerif and Tamarin. However, we ran into difficulties with long-term mutable state (which is always a challenge for infinite state verification) and since handling mutable state is an essential feature of choreographies, we opted for translation directly to ProVerif. Indeed our aim is similar to SAPIC+ to give modelers a way to address several methods without deep technical knowledge of those, thus we will investigate as part of future work how to connect more methods, possibly using SAPIC+.} Fu\added{rther fu}ture work includes widening the class of supported equational theories, \added{and} improving automation for non-monotonic state encodings\deleted{, and adding additional backend targets}.


\bibliographystyle{IEEEtran} 
\bibliography{biblio}

\appendix

\extver{
\subsection{Exporting to ProVerif (extended)}\label{app-sec:proverif}
\paragraph{ProVerif encoding}
In this section, we will describe how to use the projection semantics from Section~\ref{sec:projection} to 
automatically convert a choreography into a ProVerif process. For the rest of this section, we assume the following:
\begin{itemize}
    \item We have a choreography $\C$
    \item In $\C$ there are roles $\role A_1, \ldots,\role A_{n_{\role A}}$
    \item The initial knowledges are given in $F_{\role A_1}, \ldots,F_{\role A_{n_{\role A}}}$
    \item We have projected the choreography for each role:  \\
    $\LC_{\role A_1} = F_{\role A_1}(\sem{F_{\role A_1}:\C}_{\role A_1}),\ldots,\LC_{n_{\role A}} = F_{n_{\role A}}(\sem{F_{n_{\role A}}:\C}_{n_{\role A}})$
    \unfix{simon}{I don't think in this list is the right place to go into why you have to apply the frame after projection, but it should maybe be mentioned somewhere}
    \item Of all the roles,  $\role U_1,\ldots,\role U_{n_{\role U}}$ are untrusted.
    \item For each role $\role A_i$, $\role U_{\role A_i 1},\ldots,\role U_{\role A_i n_{\role U_{\role A_i}}}$ are the untrusted 
    roles occurring in $\LC_{\role A_i}$ and $F_{\role A_i}$ (i.e. the parameters to the role)
    \item In $\C$ we have the cell families $c_1,\ldots,c_{n_c}$
\end{itemize}

We start by encoding the algebra. 
For the most part, we encode the algebra $E$ as is, though a simplification must be made with regard to $\exp$.
Directly inserting the equation in $B$ in ProVerif will lead to nontermination. 
However, as shown in \cite{DBLP:conf/csfw/KustersT09} and \cite{DBLP:conf/ifip1-7/Modersheim11},
it is sound to use the following encoding for standard Diffie-Hellman:
\begin{tabbing}
$\kwl{fun}\ \kwf{exp}(\kwt{bitstring}, \kwt{bitstring}){:}\ \kwt{bitstring}. $\\
$\kwl{const}\ \kwc{g}{:}\ \kwt{bitstring}. $\\
$\kwl{equation}\ \kwl{forall}\ \var{x}{:}\ \kwt{bitstring}, \var{y}{:}\ \kwt{bitstring}; $\\
$\ \ \ \ \kwf{exp}(\kwf{exp}(\kwc{g}, \var{y}), \var{x}) = \kwf{exp}(\kwf{exp}(\kwc{g}, \var{x}), \var{y}). $
\end{tabbing}

We will omit most of the definition of the algebra here, but include the encoding of symmetric encryption: 
\begin{tabbing}
$\kwl{type}\ \kwt{skey}. $\\
$\kwl{fun}\ \kwf{scrypt}(\kwt{bitstring}, \kwt{skey}){:}\ \kwt{bitstring}. $\\
$\kwl{reduc}\ \kwl{forall}\ \var{x}{:}\ \kwt{bitstring}, \var{k}{:}\ \kwt{skey}; $\\
$\ \ \ \ \kwf{dscrypt}(\kwf{scrypt}(\var{x}, \var{k}), \var{k}) = \var{x}. $\\
$\kwl{fun}\ \kwf{vscrypt}(\kwt{bitstring}, \kwt{skey}){:}\ \kwt{bool} $\\
$\kwl{reduc}\ \kwl{forall}\ \var{x}{:}\ \kwt{bitstring}, \var{k}{:}\ \kwt{skey};\ $\\
$\ \ \ \ \kwf{vscrypt}(\kwf{scrypt}(\var{x}, \var{k}), \var{k}) = \kwc{true} $\\
$\kwl{otherwise}\ \kwl{forall}\ \var{x}{:}\ \kwt{bitstring}, \var{k}{:}\ \kwt{skey};\ $\\
$\ \ \ \ \kwf{vscrypt}(\var{x}, \var{k}) = \kwc{false}. $
\end{tabbing}

As another preliminary step, we include the type of agents and the name of the intruder:\\
$\kwl{type}\ \kwt{agent}. $\\
$\kwl{free}\ \kwc{i}{:}\ \kwt{agent}. $

ProVerif processes are written in the \emph{$\pi$-calculus}, and look much like our local behaviors, 
except for network communication, nondeterministic choice, and memory cells. 
We write $\pvt{\LC}$ for the ProVerif process obtained from the local behavior $\LC$, 
and in the following only describe the cases where the local behaviors differ from the resulting ProVerif processes. 

In our local behaviors, we simply use $\mathsf{send}(X)$ and $\mathsf{receive}(X)$ to send the content of variable $X$ to the public network 
or receive a value from the network and assign it to $X$. 
In ProVerif, the corresponding statements are $\kwl{out}(\kwc{c}, \var{X})$ and $\kwl{in}(\kwc{c}, X{:}\kwt{bitstring})$,
where $\kwc{c}$ is a public channel. 
To make a clear conceptual divide, we declare two public channels: $\kwc{c}$ is the public channel
the communication declared in the choreography happens over, while we use $\kwc{ic}$ when information must 
be given to or taken from the intruder for the purpose of modeling. \\
$\kwl{free}\ \kwc{ic}{:}\ \kwt{channel}. $\\
$\kwl{free}\ \kwc{c}{:}\ \kwt{channel}. $\\
We then define $\pvt{\mathsf{send}(t).\LC} = \kwl{out}(\kwc{c}, t);\pvt{\LC}$ and $\pvt{\mathsf{receive}(X).\LC} = \kwl{in}(\kwc{c}, X{:}\kwt{bitstring});\pvt{\LC}$.

When a process makes a nondeterministic choice, our encoding simply has the intruder decide:
\begin{tabbing}
$\pvt{\LC_1 + \LC_2 + \ldots + \LC_n}=$\\
$\ \ \kwl{in}(\kwc{ic}, \var{Branch}{:}\var{nat}); $\\
$\ \ \kwl{if}\ \var{Branch} = 0\ \kwl{then}\  $\\
$\ \ \ \ \ \ \pvt{\LC_1} $\\
$\ \ \kwl{else}\ \kwl{if}\ \var{Branch} = 1\ \kwl{then}\  $\\
$\ \ \ \ \ \ \pvt{\LC_2}\  $\\
$\ \ \ \ \ \ \ \ \vdots$\\
$\ \ \kwl{else}\ \kwl{if}\ \var{Branch} = \var{n}\ \kwl{then}\  $\\
$\ \ \ \ \ \ \pvt{\LC_n}$
\end{tabbing}

With regard to memory cells, we restrict ourselves to choreographies with the following properties:
\begin{itemize}
    \item In every atomic section there is at most one read and one write to each cell family on every branch, 
    and if a branch contains both they must use the same address.
    \item All writes happen at the end of the atomic section. 
    \item All addresses are public. \footnote{This is not technically required for the translation to be sound. However, our translation will reveal addresses to the intruder, so you will get false attacks if this is not the case. }
\end{itemize}

For each cell family, $c_i$, we declare a function:\\
$\kwl{fun}\ \kwf{cell{\_}}c_i(\kwt{bitstring}){:}\ \kwt{channel}\ [\var{private}]. $\\
We can use this function to create a private channel for each cell in the family. 
Each such channel will contain at most one message at any given time, representing the 
current value stored in the cell. 

Information sent on a private channel is of course not revealed to the intruder.
Furthermore, communication over private channels is synchronous,
and we can use this to enforce atomicity when there is a read and write to the 
same cell in one atomic section. 

The same trick can be used to create atomic sections in general.
We define \\
$\kwl{free}\ \kwc{atomic{\_}lock}{:}\ \kwt{channel}\ [\var{private}]. $\\
$\kwl{const}\ \kwc{atomic{\_}baton}{:}\ \kwt{bitstring}.\  $\\
To enter an atomic section, a process must obtain the baton, 
and they should send it back when they leave.
We have not found it beneficial to enforce atomicity of all atomic sections from local behaviors,
but use it strategically for some parts (for example the memory-initializer processes below).

We ensure that when an agent wants to read from a cell there is always a message available
by including initializer processes for each cell family. ProVerif supports \emph{tables}, 
which are set-like structure we can use to keep track of which cells have been initialized,
thus making sure that we only initialize a cell once and enforce the property that each cell channel contains at most one message.
For each $c_i$, we define:
\begin{tabbing}
$\kwl{table}\ c_i\kwt{{\_}initializer{\_}table}(\kwt{bitstring}). $\\
$\kwl{let}\ c_i\var{{\_}initializer}() =  $\\
$\ \ \ \ \kwl{in}(\kwc{atomic{\_}lock},  = \kwc{atomic{\_}baton}); $\\
$\ \ \ \ \kwl{in}(\kwc{ic}, \var{Addr}{:}\ \kwt{bitstring}); $\\
$\ \ \ \ \kwl{get}\ c_i\kwt{{\_}initializer{\_}table}( = \var{Addr})\ \kwl{in}\ 0 $\\
$\ \ \ \ \kwl{else} $\\
$\ \ \ \ \ \ \ \ \kwl{insert}\ c_i\kwt{{\_}initializer{\_}table}(\var{Addr}); $\\
$\ \ \ \ \ \ \ \ (\kwl{out}(\kwf{cell{\_}}c_i(\var{Addr}), \kwc{blank})\mid $\\
$\ \ \ \ \ \ \ \ \ \kwl{out}(\kwc{atomic{\_}lock}, \kwc{atomic{\_}baton})). $
\end{tabbing}
The outputs at the end must be put in parallel, so that one does not block the other. 

We can now have the intruder initialize our memory cells: \\
$\pvt{X := c[t].\LC} = \kwl{out}(\kwc{ic}, t); \kwl{in}(\kwf{cell{\_}}c(t), X{:}\ \kwt{bitstring});\pvt{\LC}$ \\
Also, writing should not block the rest of the process:\\
$\pvt{c[t] := s.\LC} = (\kwl{out}(\kwf{cell{\_}}c(t), s)\mid\pvt{\LC})$.

There are two remaining problems: If on a branch in an atomic section 
there is a read from a cell but no write, the value in the cell will be consumed 
and no other process will be able to read from it. 
Similarly, if there is a write but no read, the cell will contain two values, 
and either is a possible value when the next process reads from the cell (we also lose the atomicity mentioned above). 
We solve this by a simple transformation:
If we are done translating a branch of an atomic section where we have processed $X := c[t]$ but no corresponding write and need 
to return the rest of the translation $\pvt{\LC}$, return instead $(\kwl{out}(\kwf{cell{\_}}c(t), s)\mid\pvt{\LC})$.
Similarly, if we want to translate $\pvt{c_1[t_1] := s_1.\ldots.c_n[t_n] := s_n.\LC}$ (recall that writes must be at the end of an atomic section)
and have not read from $c_i$, we return $\kwl{out}(\kwc{ic}, t_i); \kwl{in}(\kwf{cell{\_}}c_i(t_i), \_{:}\ \kwt{bitstring}); \pvt{c_1[t_1] := s_1.\ldots.c_n[t_n] := s_n.\LC}$.
\unfix{simon}{should I make this precise by adding more parameters to $\pvt{\LC}$ and all that?}

We now how a way to translate local behaviors to ProVerif processes. However, we still need to compose all these into a single process that can be verified by ProVerif. 
Furthermore, this process must handle the multiple sessions and instantiations supported by the semantics in Figure~\ref{fig:semantics:local}. 

For each role $\role A_i$ with parameters $\role U_{\role A_i 1},\ldots,\role U_{\role A_i n_{\role U_{\role A_i}}}$ we define the \emph{process spawner} for that agent:
\begin{tabbing}
$\kwl{let}\ \var{process}{\role A_i}({\role U_{\role A_i 1}}{:}\ \kwt{agent}, \ldots, {\role U_{\role A_i n_{\role U_{\role A_i}}}}{:}\ \kwt{agent}) = \pvt{\LC_{\role A_i}}. $\\
$\kwl{let}\ \var{spawn}{\role A_i}() =  $\\
$\ \ \ \ \kwl{in}(\kwc{ic}, ({\role U_{\role A_i 1}}{:}\ \kwt{agent}, \ldots, {\role U_{\role A_i n_{\role U_{\role A_i}}}}{:}\ \kwt{agent})); $\\
$\ \ \ \ \var{process}{\role A_i}({\role U_{\role A_i 1}}, \ldots, {\role U_{\role A_i n_{\role U_{\role A_i}}}}). $
\end{tabbing} 
For each untrusted $\role U_i$ we must also give the intruder the associated initial knowledge.
Let $\role U_{\role U_i 1},\ldots,\role U_{\role U_i n}$ be all the parameters of the role $\role U_i$ except the role $\role U_i$ itself and $t_1,\ldots,t_m$ be the terms in the knowledge of $\role U_i$ 
(i.e. $\F_{\role U_i} = \{\X_1 \mapsto t_1, \ldots,\X_m \mapsto t_m\}$). We define:
\begin{tabbing}
$\kwl{let}\ \var{knowledge}{\role U_i}() =  $\\
$\ \ \ \ \kwl{in}(\kwc{ic}, ({\role U_{\role U_i 1}}{:}\ \kwt{agent}, \ldots, {\role U_{\role U_i n}}{:}\ \kwt{agent})); $\\
$\ \ \ \ \kwl{out}(\kwc{ic}, (t_1[\role i/\role U_i],\ldots,t_m[\role i/\role U_i])). $
\end{tabbing}

Finally, we can define the main ProVerif process: 
\begin{tabbing}
$\kwl{process} $\\
$\ \ \ \ (\kwl{out}(\kwc{atomic{\_}lock}, \kwc{atomic{\_}baton}))\mid $\\
$\ \ \ \ !(\kwl{new}\ \var{a}{:}\ \kwt{agent};\ \kwl{out}(\kwc{ic}, \var{a}))\mid $\\
$\ \ \ \ !(\var{spawn}{\role A_1}())\mid $\\
$\ \ \ \ \ \ \vdots$\\
$\ \ \ \ !(\var{spawn}{\role A_{n_\role A}}())\mid $\\
$\ \ \ \ !(\var{knowledge{\role U_1}}())\mid $\\
$\ \ \ \ \ \ \vdots$\\
$\ \ \ \ !(\var{knowledge{\role U_{n_\role U}}}())\mid $\\
$\ \ \ \ !c_1\var{{\_}initializer}()$ \\
$\ \ \ \ \ \ \vdots$\\
$\ \ \ \ !c_{n_c}\var{{\_}initializer}()$ \\
\end{tabbing}

\paragraph{Improved Encoding of Memory Cells}
When verifying ProVerif code produced by the encoding above, we run into the issue that the overapproximations associated with the abstractions of ProVerif make many secure protocols unverifiable. 
By default, ProVerif will consider the values in a cell as a set, and when you read you will not necessarily get the last value that was written. 
This means that for example the TPM example included with our submission does not verify, since ProVerif considers runs where $\mathit{deleted}$ is written to the cell 
and then $\mathit{classified}$ is later read from it. 

\added{This kind of overapproximation around mutable state is well-known and has been addressed by several extensions and front-ends to ProVerif and related tools \cite{Modersheim10AIF,BruniModersheimNielson15,ArapinisRitterRyan11,DBLP:conf/csfw/ChevalCT18,HessModersheimBruckerSchlichtkrull21}.}
We \replaced{improve}{can improve} the accuracy of our encoding by a trick inspired by the ``precise'' option in ProVerif, originally presented in \cite{DBLP:conf/csfw/ChevalCT18}.
When writing $t$ to a cell $c$, we instead write $(i_c,t)$ where $i_c$ is a natural number indicating that this was the $i_c$th value written to $c$.
Furthermore, we trigger an event registering what was written, every time a value is written to a call. 
The new encoding of cell $c_i$ becomes: 
\begin{tabbing}
$\kwl{fun}\ \kwf{cell{\_}}c_i(\kwt{bitstring}){:}\ \kwt{channel}\ [\var{private}]. $\\
$\kwl{event}\ \kwe{write{\_}}c_i(\kwt{bitstring}, \kwt{nat}, \kwt{bitstring}). $\\
$\kwl{table}\ c_i\kwt{{\_}initializer{\_}table}(\kwt{bitstring}). $\\
$\kwl{let}\ c_i\var{{\_}initializer}() =  $\\
$\ \ \ \ \kwl{in}(\kwc{atomic{\_}lock},  = \kwc{atomic{\_}baton}); $\\
$\ \ \ \ \kwl{in}(\kwc{ic}, \var{Addr}{:}\ \kwt{bitstring}); $\\
$\ \ \ \ \kwl{get}\ c_i\kwt{{\_}initializer{\_}table}( = \var{Addr})\ \kwl{in}\ 0 $\\
$\ \ \ \ \kwl{else} $\\
$\ \ \ \ \ \ \ \ \kwl{event}\ \kwe{write{\_}}c_i(\var{Addr}, 0, \kwc{blank}); $\\
$\ \ \ \ \ \ \ \ \kwl{insert}\ c_i\kwt{{\_}initializer{\_}table}(\var{Addr}); $\\
$\ \ \ \ \ \ \ \ (\kwl{out}(\kwf{cell{\_}}c_i(\var{Addr}), (0,\kwc{blank}))\mid $\\
$\ \ \ \ \ \ \ \ \ \kwl{out}(\kwc{atomic{\_}lock}, \kwc{atomic{\_}baton})). $
\end{tabbing}
Furthermore, we redefine the process translation with the following: \\
$\pvt{X := c[t].\LC} = 
\begin{array}{l}
    \kwl{out}(\kwc{ic}, t); \\
    \kwl{in}(\kwf{cell{\_}}c(t), (\var{Counter}_c{:}\ \kwt{nat},X{:}\ \kwt{bitstring}));\\
\pvt{\LC}\end{array}$ \\
$\pvt{c[t] := s.\LC} = 
\begin{array}{l}
    \kwl{event}\ \kwe{write{\_}}c(t, \var{Counter}_c, s); \\
    (\kwl{out}(\kwf{cell{\_}}c(t), (\var{Counter}_c + 1,s))\mid\pvt{\LC})
\end{array}$\\

When inserting writes on branches that only have reads, we can omit incrementing the counter, 
as we will necessarily just write back the value that was already there. 
In our experiments, this omission increases the efficiency of ProVerif in many cases.

Our encoding ensures that every time a new value is put in the cell, it is associated with a new counter value. 
Thus, it is sound to add the following axiom: 
\begin{tabbing}
$\kwl{axiom}\ t{:}\ \kwt{bitstring}, \var{C}{:}\ \var{nat}, s_1{:}\ \kwt{bitstring}, s_2{:}\ \kwt{bitstring}; $\\
$\ \ \ \ \kwl{event}(\kwe{write{\_}}c_i(t, \var{C}, s_1)) \wedge \kwl{event}(\kwe{write{\_}}c_i(t, \var{C}, s_2))$\\
$\ \ \ \ \ \ \ \ \Longrightarrow s_1 = s_2. $
\end{tabbing}
}

\extver{
\subsection{Automation of the Projection in Detail}\label{app:mech}

We will here show in more detail how to solve the \textit{word problem}, \textit{recipe-composition problem}, and \textit{complete-set-of-checks problem} 
described in Section~\ref{sec:mec}.

In addition to the $\Sigma$ and $\Sigma_p$ from Section~\ref{sec:mec}, we use $\Sigma_d$ to denote the set of destructors (+ pair projections, Diffie-Hellman inverse, etc.) and $\Sigma_v$ for the set of verifiers. 

Unlike in the main matter of the paper, we will here permit terms $s,t,\ldots$ to contain all function symbols, 
as we here make precise how these terms are then reduced to ones built only from $\Sigma$.

For the benefit of the following proofs, we define the syntactic size of terms and checks:
$|\X| = 1$, $|f(t_1,\ldots,t_n)| = 1 + |t_1| + \ldots + |t_n|$, and $|t \doteq s| = |t| + |s|$.

Additionally, we need the following concepts: 
\begin{definition}
    We define $[t]_{= B} = \{t' \mid t =_B t'\}$.
    We write $\rightarrow_R$ for the rewriting system obtain by applying the equations in $R$ from left to right. 
    We can then define the rewriting system modulo $B$-equivalence classes: 
    $s \rightarrow_{R/B} t$ iff $\exists s'\ t'.\ s =_B s' \land s' \rightarrow_R t' \land t' =_B t$.
    We use $t\downarrow_{R/B}$ to denote the normal form of $t$. 
\end{definition}

\begin{lemma}
  If $s =_B t$ then $|s| = |t|$ and if $s \rightarrow_{R/B} t$ then $|s| > |t|$.
\end{lemma}
\begin{proof}
  The first part can be shown by induction on the derivation of $s =_B t$,
  and the second part follows by case analysis on $s' \rightarrow_R t'$ after 
  unfolding the definition of $\rightarrow_{R/B}$.
\end{proof}

\begin{lemma}\label{lemma:algeq-split}
  $s \algeq t$ if and only if either:
  \begin{itemize}
    \item $s =_B t$,
    \item or there is $s'$ so $s \rightarrow_{R/B} s'$ and $s' \algeq t$,
    \item or there is $t'$ so $t \rightarrow_{R/B} t'$ and $s \algeq t'$
  \end{itemize}
\end{lemma}
\begin{proof}
  We do induction on the derivation of $s \algeq t$ in the equational logic induced by $E$.
\end{proof}

We can now show that our rewriting system is well-behaved:
\begin{lemma}\label{lemma:convergence}
    (A) $[t]_{= B}$ is finite for all $t$. 
    (B) If $s \rightarrow_{R/B} t$ then $s \algeq t$.
    (C) $\rightarrow_{R/B}$ is convergent, modulo $B$.
    (D) $t\downarrow_{R/B}$ is defined and unique for all $t$, modulo $B$.
\end{lemma}
\begin{proof}
  (A) is easy to see (the equation in $B$ is a kind of commutativity). 
  (B) follows from transitivity of $\algeq$ and the definitions of $\rightarrow_R$ and $\rightarrow_{R/B}$.
  (D) follows (C) by definition. 

  To show (C), we combine (B) with the following: 
  (1) $\rightarrow_{R/B}$ has no infinite chains. 
  (2) If $s \algeq t$ and $\rightarrow_{R/B}$ applies to neither $s$ nor $t$, then $s =_B t$.

  (1) follows from the fact that if $t_1 \rightarrow_{R/B} t_2 \rightarrow_{R/B} \ldots$ then $|t_1| > |t_2| > \ldots$.

  (2) follows from Lemma~\ref{lemma:algeq-split}.
\end{proof}

With this, we can solve the \textit{word problem} from the start of the section:
\begin{thm}\label{thm:word:problem}
    $s =_E t$ if and only if $s\downarrow_{R/B} =_B t\downarrow_{R/B}$.
\end{thm}
\begin{proof}
  That $s\downarrow_{R/B} =_B t\downarrow_{R/B}$ implies $s =_E t$ follows from Lemma~\ref{lemma:convergence} (B).

  If $s =_E t$ then $s\downarrow_{R/B} =_E t\downarrow_{R/B}$ by Lemma~\ref{lemma:convergence} (B),
  and since $\rightarrow_{R/B}$ applies to neither $s\downarrow_{R/B}$ nor $t\downarrow_{R/B}$ we 
  have $s\downarrow_{R/B} =_B t\downarrow_{R/B}$ by Lemma~\ref{lemma:algeq-split}.
\end{proof}

In the following, we will analyze a set of frames by alternating between applying verifiers and destructors to obtain new terms. 
We extend frames with additional information, including which checks have been performed, and how the content in a label was derived. 

\begin{definition}[Enhanced Frames]
  Frames are built over the following grammar: 
  \begin{eqnarray*}
    F &::=& 0\\
      &\mid& F_0.\X \mapsto t \\
      &\mid& F_0.\X\leftarrow r \mapsto t\\
      &\mid& F_0.r_1\doteq r_2
  \end{eqnarray*}
  We let $e$ range over the entries of the form $\X \mapsto t$, $\X\leftarrow r \mapsto t$, and $r_1\doteq r_2$.
  \begin{itemize}
  \item The \emph{domain} of a frame is defined as follows. $\dom(0)=\emptyset$, $\dom(F_0.\X\mapsto
    t)=\dom(F_0)\cup\{\X\}$, $\dom(F_0.\X\leftarrow r\mapsto
    t)=\dom(F_0)\cup\{\X\}$, and $\dom(F_0. r_1\doteq
    r_2)=\dom(F_0)$.
  \item Given a recipe $r$ for frame $F$, we define $F(r)$ as expected, ignoring $r_1\doteq r_2$ and $\leftarrow r$.
  \end{itemize}
  In the following, we only consider \emph{well-formed frames}, which
  satisfy these properties:
  \begin{itemize}
  \item In a frame $F_0.\X\mapsto t$, $\X\notin\dom(F_0)$ and $t$ is constructive. 
  \item In a frame $F_0.\X\leftarrow r\mapsto t$, $\X\notin\dom(F_0)$,
    $\fv(r)\subseteq\dom(F_0)$, $F_0(r) = t$, and there exists constructive $t'$ such that $t \algeq t'$.
  \item In a frame $F_0.r_1\doteq r_2$,
    $\fv(r_1)\cup\fv(r_2)\subseteq\dom(F_0)$, and $F_0(r_1)\algeq
    F_0(r_2)$.
  \item In a frame $F_0.e$, $F_0$ is well-formed. 
  \end{itemize}
  It is easy to see that the projection semantics above and the following procedures preserve well-formedness.

  The \emph{checks} contained in a frame
  $F$, written $\mathit{checks}(F)$, is the set of all the equations
  $r_1 \doteq r_2$ contained in the frame together with 
  the equations $l \doteq r$ for each entry $l \gets r \mapsto t$ in the frame. 
\end{definition}

In the projection semantics, we consider sets of frames. There is a notion that these frames are identical with regard to 
all operations we have performed so far, which we formalize in the following:
\begin{definition}[Compatible Frames]
  We define two frames being \emph{compatible}, written
  $F\simeq F'$, as the least relation satisfying the following rules:
  \begin{displaymath}
    \begin{array}{cccccc}
      \infer{0\simeq 0}{}\\ 
      \infer{F_0.\X\mapsto t\simeq  F_0'.\X\mapsto t'}{F_0\simeq
      F_0'}\\
      \infer{F_0.\X\leftarrow r\mapsto t\simeq  F_0'.\X\leftarrow r\mapsto t'}{F_0\simeq
      F_0'}\\
      \infer{F_0.r_1\doteq r_2\simeq F_0'.r_1\doteq r_2}{F_0\simeq F_0'}
    \end{array}
  \end{displaymath}
\end{definition}

Thus, frames are compatible if they have passed the same checks and performed the same operations to construct new terms, in the same order.

Our analysis should extend a frame to contain all derivable subterms. This goal is captured in the following: 
\begin{definition}
  A label $\X \in \dom(F)$ is \emph{analyzed} (in $F$) if for any constructive key-recipe $r_k$ and destructor $d \in \Sigma_d$,
  if $d$ applies to $(F(\X),F(r_k))$, there exists an entry $\X' \leftarrow d(\X,r_k') \mapsto t$ where $F(r_k) \algeq F(r_k')$.
  Furthermore, for the associated verifier $v$ there must be an entry $v(\X,r_k') \doteq \top$. \\
  A frame $F$ is \emph{fully analyzed} if all labels in $\dom(F)$ are analyzed.  
\end{definition}
Note that since $F$ is well-formed, we must have $t \algeq d(F(\X),F(r_k))$.

When all the subterms have been added to a frame, we can obtain any obtainable term by \emph{composition},
i.e. a constructive recipe. The following recursive function returns all constructive recipes for a given term: 
\begin{definition}
    We define $\mathit{compose}(F,t) = \bigcup\{\mathit{compose}_{lc}(F,s) \mid s \in [t]_{= B}\}$ 
    where \\
    $\mathit{compose}_{lc}(F,s) = \mathit{compose}_l(F,s) \cup \mathit{compose}_c(F,s)$, \\
    $\mathit{compose}_l(F,s) = \{\X | \X \in \dom(F) \land F(\X) = s\}$, and \\
    $\mathit{compose}_c(F,s) = \{f(s_1',\ldots,s_n') \mid s_1' \in \mathit{compose}(F,s_1) \land \ldots \land s_n' \in \mathit{compose}(F,s_n)\}$ if $s = f(s_1,\ldots,s_n)$ for $f \in \Sigma_c$ \\ 
    and $\mathit{compose}_c(F,s) = \{\}$ otherwise. 
\end{definition}
That $\mathit{compose}(F,t)$ is finite can be seen by induction on $|t|$ and that $[t]_{= B}$ is finite.

The following shows that for fully analyzed frames, $\mathit{compose}$ can create any obtainable term:
\begin{lemma}\label{lemma:compose-in-analyzed}
  Assume that $F$ is a fully analyzed frame. \\
  If there is a recipe $r$ and a constructive term $t$ such that $F(r) \algeq t$, 
  then we have that \\
  (1) $\mathit{compose}(F,t)$ is non-empty, and\\
  (2) for every $r_c \in \mathit{compose}(F,t)$ we have $F(r_c) \algeq t$.
\end{lemma}
\begin{proof}
  (1)
  We show that if $F(r) \rightarrow_{R/B}^* t$ then there is $r_c \in \mathit{compose}(F,t)$ such that $F(r_c) \rightarrow_{R/B}^* t$.\\
  We do induction on $|r|$. 
  If $r = \X$, we are done. 
  If $r = f(r_1,\ldots,r_n)$ for $f \in \Sigma_c$, we apply the induction hypothesis and are done. 
  If $r = d(r_1,\ldots,r_n)$ for $f \in \Sigma_d$, then either the recipe reduces and we are done by the induction hypothesis, or 
  $r = d(\X,r_\mathit{key})$ and there is another label $\X'$ and (by the induction hypothesis) a constructive recipe $r_\mathit{key}'$ such that $F(\X') = F(d(\X,r_\mathit{key}'))$.
  \unfix{simon}{this proof might need closer inspection}
  
  (2) It is easy to prove that for every $r_c \in \mathit{compose}_\mathit{lc}(F,t)$ we have $F(r_c) \algeq t$ by induction on $|t|$. 
  Furthermore, it is by definition the case that for every $s \in [t]_{= B}$ we have $s \algeq t$.
\end{proof}

The following shows completeness of $\mathit{compose}$:
\begin{lemma}\label{lemma:compose-complete}
  For any frame $F$ (even if not fully analyzed), if $r$ is a constructive recipe such that $F(r) \algeq t$ and $t$ is constructive, then $r \in \mathit{compose}(F,t)$.
\end{lemma}
\begin{proof}
  Follows by induction on $|r|$. 
\end{proof}

In the following, we define an analysis procedure that can be triggered 
during the projection semantics of Section~\ref{sec:projection} any time a new term is added the frame. 

For the analysis of frame $F$, we require a \emph{marking} of the type $M: \dom(F) \longrightarrow \{\star,+,\checkmark\}$.
All labels are initially marked $\star$,
and whenever a label is added to the frame 
it is also marked $\star$. 

\begin{definition}
  A marking $M$ of frame $F$ is \emph{accurate} if any 
  label $\X \in \dom(F)$ where $M(X) = \checkmark$ is analyzed in $F$.
\end{definition}

\begin{definition}[Analysis Procedure]
  We want to calculate the \emph{analyzed extensions} of a set of frames, $\{F_1,\ldots,F_n\}$, and marking, $M$, 
  where the marking is accurate for each frame and the frames are pairwise compatible. 
  The analysis-extension function, $\Ana(M,\{F_1,\ldots,F_n\})$, is defined by three cases: \\
  If the given set is empty ($n = 0$), return $\{\}$. \\
  If there is a label $\X$ marked $\star$, try the following in order:
  \begin{itemize}
    \item $F_1(\X),\ldots,F_n(\X)$ are all terms for which no verifier exists, then return $\Ana(M[\X \mapsto \checkmark],\{F_1,\ldots,F_n\})$.
    \item (\textit{decomposition}) We have a frame $F_i$, a verifier $v \in \Sigma_v$, and a constructive term $t$ such that $v$ applies to $(F_i(\X),t)$ (for any $d$ and $F_i(\X)$ at most one such $t$ exists and it is easy to find).
    Furthermore, $\mathit{compose}(F_i,t)$ is non-empty, containing at least $r_{\mathit{key}}$.  \\
    Then, let $d_1,\ldots,d_k$ be all destructors associated with $v$ and $\X_1,\ldots,\X_k$ fresh labels,
    and set $r_1 = d_1(\X,r_{\mathit{key}}),\ldots,r_k = d_k(\X,r_{\mathit{key}})$. \\
    Let $M'$ be the marking that for any $\X'$ returns $\checkmark$ if $\X = \X'$, $\star$ if $\X' \in \{\X_1,\ldots,\X_k\}$ or $M(\X') = +$, and $M(\X')$ otherwise. \\
    Finally, let \\ 
    $\mathit{FS}^+ = \{F.v(\X,r_{\mathit{key}}) \doteq \top.X_1 \gets r_1 \mapsto F(r_1).\, \cdots\, .X_k \gets r_k \mapsto F(r_k) \mid F \in \{F_1,\ldots,F_n\} \land F(v(\X,r_{\mathit{key}})) \algeq \top\}$\\ and \\
    $\mathit{FS}^- = \{F  \mid F \in \{F_1,\ldots,F_n\} \land F(v(\X,r_{\mathit{key}})) \not\algeq \top\}$.\\
    We then return $\Ana(M',\mathit{FS}^+) \cup \Ana(M,\mathit{FS}^-)$.
    \item Return $\Ana(M[\X \mapsto +],\{F_1,\ldots,F_n\})$.
  \end{itemize}
  If no labels are marked $\star$, return $\{(M,\{F_1,\ldots,F_n\})\}$.
\end{definition}

Note that the choice of $r_\mathit{key}$ does not matter, since if multiple recipes are supposed to 
produce the same term but actually do not, we are not in a run of the protocol that agrees with the given frame.
This will then be discovered when we later do the checks. \unfix{simon}{I'm wondering if this should be stated more formally in some way}

\begin{thm}\label{thm:ana:terminate}
  The analysis procedure terminates.
\end{thm}
\begin{proof}
    Let $\mathit{raw}(F)$ be all the entries of the form $\X \mapsto t$ (i.e. the label assignments that are not the product of analysis). 
    Let $|F| = \sum \{|t| \mid \X \mapsto t \in \mathit{raw}(F)\}$ and $|\{F_1,\ldots,F_n\}| = \sum_{i=1}^{n} |F_i|$.
    Let $|M|^\checkmark$ be the number of labels marked $\checkmark$ in $M$ and $|M|^+$ be the number of labels marked $+$. 
    Lexicographically measure $(|\{F_1,\ldots,F_n\}| - |M|^\checkmark,|\{F_1,\ldots,F_n\}| - |M|^+)$.\\
    The number of labels in a marking will never exceed the number of subterms of the raw terms of the frames the marking is based on,
    so $|\{F_1,\ldots,F_n\}| - |M|^\checkmark$ and $|\{F_1,\ldots,F_n\}| - |M|^+$ will always be at least 0.\\
    It is easy to see that no step in the analysis makes $|\{F_1,\ldots,F_n\}| - |M|^\checkmark$ increase, and that 
    each step will either decrease $|\{F_1,\ldots,F_n\}| - |M|^\checkmark$ or $|\{F_1,\ldots,F_n\}| - |M|^+$.
\end{proof}

\begin{thm}\label{thm:ana:correct}
  If $F_1,\ldots,F_n$ are compatible frames, $M$ is an accurate marking, and $\Ana(M,\{F_1,\ldots,F_n\}) = \{(M_1,\mathit{FS}_1),\ldots,(M_m,\mathit{FS}_m)\}$,
  then all frames in each $\mathit{FS}_i$ are fully analyzed, pairwise compatible, and $M_1,\ldots,M_m$ are accurate.
\end{thm}
\begin{proof}
  It is easy to see that the procedure preserves compatibility and accuracy of markings.

  When the procedure is done, no label will be marked $*$.
  If one of $F_{ij}$ is then not fully analyzed, there must be a label $\X$, a constructive term $t$, and (by Lemma~\ref{lemma:compose-complete}) a recipe $r_k \in \mathit{compose}(F,t)$,
  such that the \textit{(decomposition)} case is applicable for $\X$ (which must mean, by accuracy, that $\X$ is marked $+$). 
  However, this cannot be the case, as if this $r_k$ existed when the marking of $\X$ was changed from $\star$ to $+$ then the wrong case was applied,
  and if it did not exist then the frame must have changed since this, which would remark $\X$ as $\star$. 
\end{proof}

After distinguishing and extending a set of frames with this procedure, defining procedures for solving the two 
remaining problems stated at the beginning of this section becomes easier. 

We start by showing that in an analyzed frame, all checks can be done in a constructive way: 
\begin{lemma}\label{lemma:constructive-checks}
  If $F$ is a fully analyzed frame, then for any check $r_a \doteq r_b$ where $F(r_a) \algeq F(r_b)$ there 
  exists a set of constructive checks $\phi_c$
    such that 
  $\bigwedge \mathit{checks}(F) \land \bigwedge \phi_c \models_E r_a \doteq r_b$.

  If $r_a$ is of the form $f(r_a',r_k)$ where $f \in \Sigma_v \cup \Sigma_d$ and $r_a'$, $r_k$ and $r_b$ are constructive, 
  then there is such a $\phi_c$ where for each check $r_1 \doteq r_2 \in \phi_c$ 
  we have that $|r_1 \doteq r_2| \leq \mathit{max}(2*|r_k|,|r_a'| + |r_b|)$.
\end{lemma}
\begin{proof}
  We prove the first part by induction on $|r_a \doteq r_b|$. 

  The second part of the lemma can be proven by considering the cases of $r_a'$. 
  First, we note that the proof is immediate if it is not the case that $f$ is a destructor or verifier 
  that applies to $(F(r_a'),F(r_k))$, and we therefore only consider this case.\\
  If $r_a' = \X$ where $\X$ is a label, then we know that for the applicable verifier $v$ and destructors $d_1,\ldots,d_n$ 
  we have a constructive recipe $r_k'$ and labels $\X_1,\ldots,\X_n$ such that $\{v(\X,r_k') \doteq \top,d_1(\X,r_k') \doteq \X_1 \ldots d_n(\X,r_k') \doteq \X_n\} \subseteq \mathit{checks}(F)$.
  If $f = v$ we set $\phi_c = \{r_k \doteq r_k',r_b \doteq \top\}$ and are done. 
  If $f = d_i$ we set $\phi_c = \{r_k \doteq r_k',r_b \doteq \X_i\}$ and are done. \\
  If $r_a'$ is a composed recipe (i.e. of the form $g(r_{a1}',\ldots,r_{an}')$) then it must be the case that $f(r_a',r_k) \rightarrow_R r_a''$ for some constructive recipe $r_a''$. 
  We then set $\phi_c = \{r_a'' \doteq r_b\}$ and are done. 
\end{proof}

\begin{definition}
  For a term $F$ we define the set of all checks that hold in the frame: \\
  $\phi(F) = \{ r_1 \doteq r_2 \mid r_1\in\T(\Sigma_p,\dom(F)) \land r_2\in\T(\Sigma_p,\dom(F))
    \land F(r_1) \algeq F(r_2) \}$
  
  We also define the set of all checks between labels and recipes we can compose:\\
  $\phi_\mathit{lc}(F) = \{ \X\doteq r \mid  \X\in\dom(F) \land r\in\mathit{compose}(F,F(\X))\}$
\end{definition}

Since a frame is finite, and for any $t$ we have that $\mathit{compose}(F,t)$ is finite, it must be that $\phi_\mathit{lc}(F)$ is finite. 

We can now solve the \textit{complete set of checks} problem:

\begin{thm}\label{thm:complete-checks}
  If $F$ is fully analyzed, then for any $r_1 \doteq r_2 \in \phi(F)$ we have 
  $\bigwedge \mathit{checks}(F) \land \bigwedge \phi_\mathit{lc}(F) \models_E r_1 \doteq r_2$
\end{thm}
\begin{proof}
  Because of Lemma~\ref{lemma:constructive-checks} we can assume that $r_1$ and $r_2$ are constructive. 
  We do induction on $|r_1 \doteq r_2|$. 

  If either $r_1$ or $r_2$ is a label, we are done. 

  If $r_1 = f(r_{11},\ldots,r_{1n})$ and $r_2 = g(r_{21},\ldots,r_{2n})$, we 
  must have $f = g$. Furthermore, if $r_{11} \doteq r_{21},\ldots,r_{1n} \doteq r_{2n}$ are 
  valid checks in $F$, and we apply the induction hypothesis to get checks $\phi_1,\ldots,\phi_n$. 
  We have $\bigwedge \mathit{checks}(F) \land \bigwedge \phi_1 \land \ldots \land \bigwedge \phi_n \models_E r_1 \doteq r_2$ 
  and are done. 

  The only hard part is if $r_1 = \exp(r_{1a},r_{1b})$ and $r_2 = \exp(r_{2a},r_{2b})$ when 
  $F(r_{1a}) \not\algeq F(r_{2a})$ or $F(r_{1b}) \not\algeq F(r_{2b})$. 
  We then have either that $F(r_{1a}) \algeq \exp(s,t)$ and $t \algeq F(r_{2b})$ or 
  $F(r_{2a}) \algeq \exp(s,t)$ and $t \algeq F(r_{1b})$, for some constructive $s$ and $t$. 
  We consider the former and the proof for the latter is symmetric. \\
  We get $\expinv(F(r_{1a}),F(r_{2b})) = s$ and $\vexp(F(r_{1a}),F(r_{2b})) = \top$.
  From $\expinv(F(r_{1a}),F(r_{2b})) = s$ and Lemma~\ref{lemma:compose-in-analyzed}, 
  we get a constructive recipe $r'$ such that $F(r') = s$. 
  Observe that $|r'| < |r_{1a}|$. \\
  By Lemma~\ref{lemma:constructive-checks} we get that there must be constructive 
  checks $\phi^{r'}$ such that $\bigwedge \mathit{checks}(F) \land \bigwedge \phi^{r'} \models_E \expinv(r_{1a},r_{2b}) \doteq r' \land \vexp(r_{1a},r_{2b}) = \top$
  and for all $r_{\phi} \doteq r_\phi' \in \phi^{r'}$ we have $|r_{\phi} \doteq r_\phi'| \leq \mathit{max}(2*|r_{2b}|,|r_{1a}|+|r'|)$.
  Since $|r_{2b}| < |r_{1a}|$ we further have that for all $r_{\phi} \doteq r_\phi' \in \phi^{r'}$ we have $|r_{\phi} \doteq r_\phi'| < |r_1 \doteq r_2|$.
  The induction hypothesis then gives a set of checks $\phi^{r'}_{lc} \subseteq \phi_\mathit{lc}(F)$ 
  such that $\bigwedge \mathit{checks}(F) \land \bigwedge \phi^{r'}_{lc} \models_E \bigwedge \phi^{r'}$.\\
  Furthermore, we have that $F(\exp(r',r_{1b})) \algeq F(r_{2a})$, 
  which together with the induction hypothesis gives a set of checks $\phi^{r_{2a}}_{lc} \subseteq \phi_\mathit{lc}(F)$ 
  such that $\bigwedge \mathit{checks}(F) \land \bigwedge \phi^{r_{2a}}_{lc} \models_E \exp(r',r_{1b}) \doteq r_{2a}$.

  \noindent
  Combining the above with \\
  $\expinv(r_{1a},r_{2b}) \doteq r' \land \vexp(r_{1a},r_{2b}) = \top \models_E r_{1a}\doteq \exp(r',r_{2b})$ \\
  and \\
  $r_{1a}\doteq \exp(r',r_{2b}) \land \exp(r',r_{1b}) \doteq r_{2a} \models_E \exp(r_{1a},r_{1b}) \doteq \exp(r_{2a},r_{2b})$ \\ 
  yields \\
  $\bigwedge \mathit{checks}(F) \land \bigwedge \phi^{r'}_{lc} \land \bigwedge \phi^{r_{2a}}_{lc} \models_E r_1 \doteq r_2$ \\
  and we are done. 
\end{proof}

It still remains to solve \textit{recipe composition}. We have already shown that $\mathit{compose}$ will compose a recipe 
for a term in a single fully analyzed frame if it exists, but it remains to be shown that we can use it to compose a 
recipe that will work for all given frames. To prove this, we will leverage the complete set of checks. 

\begin{definition}
    A frame $F$ is \emph{fully checked} if for any $r_1 \doteq r_2 \in \phi(F)$ we have $\mathit{checks}(F) \models_E r_1 \doteq r_2$.
\end{definition}

The following checking procedure simply mimics the one from Section~\ref{sec:projection}:
\begin{definition}[Checking Procedure]
  We want to calculate the \emph{checked extensions} of a set of frames, $\{F_1,\ldots,F_n\}$ where the frames are pairwise compatible. 
  This is given by $\Chck(\phi_\mathit{lc}(F_1) \cup \ldots \cup \phi_\mathit{lc}(F_n),\{F_1,\ldots,F_n\})$ where $\Chck$ is the \emph{check-extension function}.
  $\Chck(\phi,\{F_1,\ldots,F_n\})$ is defined by three cases: \\
  If $\phi = \{\}$, return $\{\{F_1,\ldots,F_n\}\}$. \\ 
  If $n = 0$, return $\{\}$. \\ 
  If $r_1 \doteq r_2 \in \phi$, let 
    $\phi' = \phi \setminus r_1 \doteq r_2$,
    $\mathit{FS}^+ = \{F.r_1 \doteq r_2 \mid F \in \{F_1,\ldots,F_n\} \land F(r_1) \algeq F(r_2)\}$\\ and 
    $\mathit{FS}^- = \{F  \mid  F \in \{F_1,\ldots,F_n\} \land F(r_1) \not\algeq F(r_2)\}$,\\
    and then return $\Chck(\phi',\mathit{FS}^+) \cup \Chck(\phi',\mathit{FS}^-)$
\end{definition}
 
This procedure is sound:
\begin{lemma}
    If $\{F_1,\ldots,F_n\}$ is a set of compatible and analyzed frames, 
    then every $\mathit{FS} \in \Chck(\phi_\mathit{lc}(F_1) \cup \ldots \cup \phi_\mathit{lc}(F_n),\{F_1,\ldots,F_n\})$
    will contain compatible, analyzed and checked frames.  
\end{lemma}
\begin{proof}
  It is easy to see that compatibility is preserved by the procedure, and that they are still analyzed (no terms are added to the frames).

  If $F \in \mathit{FS}$ then $\phi_\mathit{lc}(F) \subseteq \mathit{checks}(F)$, and we are done by Theorem~\ref{thm:complete-checks}.
\end{proof}

As an immediate result, we get that there is no way to distinguish the frames that are grouped together by $\Chck$:
\begin{lemma}\label{lemma:check-in-one-check-in-all}
    If $F_1,\ldots,F_n$ are compatible and fully checked frames, for any $F_i$ and any check $r_1 \doteq r_2$ 
    we have that $F_i(r_1) \algeq F_i(r_2)$ if and only if $F_1(r_1) \algeq F_1(r_2) \land \ldots \land F_n(r_1) \algeq F_n(r_2)$.
\end{lemma}
\begin{proof}
  We have $F_i(r_1) \algeq F_i(r_2)$ if and only if $\mathit{checks}(F_i) \models_E r_1 \doteq r_2$.
  Furthermore, by compatibility, we have $\mathit{checks}(F_1) = \ldots = \mathit{checks}(F_n)$.
\end{proof}

We can also show a version where the frames are not necessarily compatible, but one frame fulfills all the checks in the complete set of checks of another. 
This implies that when we have done all the checks induced by the frame during the translation procedure, it does not matter which 
recipe is picked for constructing a term during actual execution (they will all refer to the same term anyways).
\begin{lemma}\label{lemma:check-in-alternative}
  If $\phi$ is a complete set of checks for $F$ and $F'$ fulfills all the checks in $\phi$, then 
  for any two recipes $r_1$ and $r_2$ such that $F(r_1) =_E F(r_2)$ we will have $F'(r_1) =_E F'(r_2)$.
\end{lemma}
\begin{proof}
  Since $r_1 \doteq r_2$ is a valid check for $F$ it must be logically implied by $\phi$, 
  which, since $F'$ satisfies $\phi$, means that it is also a valid check in $F'$.
\end{proof}

For a set of compatible, fully analyzed, and fully checked frames, the \textit{recipe composition} problem is simply solved by $\mathit{compose}$:
\begin{thm}\label{thm:compose:correct}
    If $F_1,\ldots,F_n$ are compatible, fully checked, and fully analyzed frames, $t_1,\ldots,t_n$ are constructive terms, 
    and there exists an $r$ such that $F_1(r) \algeq t_1 \land \ldots \land F_n(r) \algeq t_n$,
    then $\mathit{compose}(F_1,t_1)$ is non-empty and for any $r' \in \mathit{compose}(F_1,t_1)$ we have $F_1(r') \algeq t_1 \land \ldots \land F_n(r') \algeq t_n$.
\end{thm}
\begin{proof}
  We get the non-emptiness of $\mathit{compose}(F_1,t_1)$ and that for each $r'$ we have $F_1(r') \algeq t_1$ by Lemma~\ref{lemma:compose-in-analyzed}. 
  Since $F_1(r) \algeq F_1(r')$, we have by Lemma~\ref{lemma:check-in-one-check-in-all} that 
  $F_2(r) \algeq F_2(r') \land \ldots \land F_n(r) \algeq F_n(r')$, and are done. 
\end{proof}}

\diffver{\color{magenta}}
\subsection{A Theory with no Finite Set of Checks}\label{sec:infinite:checks}

Let $E$ be the following set of equations:

\noindent
$\{~\ord(0) \doteq \top,$ $\ord(s(x)) \doteq  \ord(x),$ $\ord(\omega) \doteq \top,$

$\lt(0,x) \doteq  \ord(x),~\lt(s(x),s(y)) \doteq  \lt(x,y),$

$\lt(s(x),\omega) \doteq 
\lt(x,\omega),~ g(x,h(y)) \doteq  \lt(x,y)~\}$.

\noindent where $g(\cdot,\cdot),0,s(\cdot)$, and $\top$ are public, and the others are
private (note that $\omega$ is also private). Let $F=[\X_1\mapsto h(\omega)]$.
Note that the corresponding term rewriting system (i.e., with
  $\ord(0)\rightarrow\top$ etc.) is convergent as can be shown with
  the critical pair methods (there are no critical pairs).

\begin{thm}
  No finite set of equations is a complete set of checks for $F$.
\end{thm}

\begin{proof}
  We assume a finite set $\phi$ of checks that are correct for $F$ (i.e.,
  $F(r)=_E F(r')$ for every $r\doteq r'\in\phi$). Without loss of
  generality, assume all checks are normalized terms.

  Let $n_0$ be the maximum $n$ such that $s^n(0)$ occurs in $\phi$, or
  $0$ if $s^n(0)$ does not occur in $\phi$. (Such an $n_0$ must exist
  since $\phi$ is finite.)

  Let $\I$ be any interpretation such that
  $\I(X_1)=h(s^{n_0+1}(0))$. Now it suffices to show that
  $\I\models\phi$, because
  $\I\not\models g(s^{n_0+1}(0),\X_1)\doteq\top$, which, since $F(g(s^{n_0+1}(0),\X_1))=_E\top$, shows that $\phi$ is not
  complete.

  Let $(r,r')\in\phi$ and we can rule out $r=r'$, and note that
  $F(r)=_EF(r')$. We give a case distinction on $r$:
  \begin{itemize}
  \item $r=X_1$ and $r=0$ is excluded because that can only be a sound
    check if $r=r'$.
  \item $r=s(r_0)$ then necessarily $r'=s(r_0')$ with
    $F(r_0)=_E F(r_0')$; thus we can just reduce this to the case of
    $r_0\doteq r_0'$.
  \item $r=g(r_n,r_m)$. Note that $r_m$ cannot be of the form
    $h(\cdot)$ because $h(\cdot)$ is private. Now case distinction on
    $r'$:
    \begin{itemize}
    \item $r'=\top$. Thus $F(r)=_E\top$, thus $F(r_m)=_E h(m)$ for
      some ground term $m$ and also $lt(F(r_n),m)=_E \top$. This
      forces $r_m=\X_1$ and $\ord(F(r_n))=_E\top$. That means $r_n$ can
      only be of the form $s^n(0)$ for some $n\leq n_0$. Thus,
      $\I\models_E r\doteq r'$.
    \item $r'=g(r_n',r_m')$. If both $F(r_n)=_E F(r_n')$ and
      $F(r_m)=F(r_m')$, then we can again reduce to the simpler checks
      $r_n\doteq r_n'$ and $r_m\doteq r_m'$.

      Otherwise, both $F(r)=_E \lt(x,y)$ and $F(r')=_E \lt(x',y')$ for
      some ground terms $x,x',y,y'$ and $\lt(x,y)=_E \lt(x',y')$.  In
      this case both $F(r_m)=_E h(m)$ and $F(r_m')=_E h(m')$ for some
      terms $m$ and $m'$, thus $r_m=r_m'=\X_1$. Since
      $F(r_n)\neq_E F(r_n')$ we have $\ord(F(r_n))=_E\top$ and
      $\ord(F(r_n'))=_E\top$. Thus, both $r_n=s^n(0)$ and $r_m=s^m(0)$
      for some $n,m\leq n_0$. Therefore, $\I\models_E r\doteq r'$.
    \item $r'$ cannot be of any other form while $F(r)=_E F(r')$.
    \end{itemize}
  \item $r=\top$. Then reduce to the case $r'\doteq r$.
  \end{itemize}
\end{proof}

\end{document}